\documentclass[a4paper,UKenglish,cleveref, autoref]{lipics-v2019}

\title{The Complexity of Packing Edge-Disjoint Paths}

\titlerunning{The Complexity of Packing Edge-Disjoint Paths}%optional, please use if title is longer than one line

\author{Jan Dreier}{Dept.\ of Computer Science, RWTH Aachen University, Germany}{dreier@cs.rwth-aachen.de}{https://orcid.org/0000-0002-2662-5303}{}%TODO mandatory, please use full name; only 1 author per \author macro; first two parameters are mandatory, other parameters can be empty. Please provide at least the name of the affiliation and the country. The full address is optional
\author{Janosch Fuchs}{Dept.\ of Computer Science, RWTH Aachen University, Germany}{fuchs@algo.rwth-aachen.de}{https://orcid.org/0000-0003-3993-222X}{}
\author{Tim A.~Hartmann}{Dept.\ of Computer Science, RWTH Aachen University, Germany}{hartmann@algo.rwth-aachen.de}{https://orcid.org/0000-0002-1028-6351}{}
\author{Philipp Kuinke}{Dept.\ of Computer Science, RWTH Aachen University, Germany}{kuinke@cs.rwth-aachen.de}{https://orcid.org/0000-0001-9716-6346}{}
\author{Peter Rossmanith}{Dept.\ of Computer Science, RWTH Aachen University, Germany}{rossmani@cs.rwth-aachen.de}{https://orcid.org/0000-0003-0177-8028}{}
\author{Bjoern Tauer}{Dept.\ of Computer Science, RWTH Aachen University, Germany}{tauer@algo.rwth-aachen.de}{}{}
\author{Hung-Lung Wang}{Computer Science and Information Engineering, National Taiwan Normal University, Taiwan}{hlwang@gapps.ntnu.edu.tw}{}{}

\authorrunning{J.\ Dreier, J.\ Fuchs, T.\ Hartmann, P.\ Kuinke, P.\ Rossmanith, B.\ Tauer and H.\-L.\ Wang}%TODO mandatory. First: Use abbreviated first/middle names. Second (only in severe cases): Use first author plus 'et al.'

\Copyright{Jan Dreier, Janosch Fuchs, Tim A.\ Hartmann, Philipp Kuinke, Peter Rossmanith, Bjoern Tauer and Hung-Lung Wang}%TODO mandatory, please use full first names. LIPIcs license is "CC-BY";  http://creativecommons.org/licenses/by/3.0/

%TODO mandatory; please add comma-separated list of keywords

\nolinenumbers %uncomment to disable line numbering

\hideLIPIcs  %uncomment to remove references to LIPIcs series (logo, DOI, ...), e.g. when preparing a pre-final version to be uploaded to arXiv or another public repository

\EventEditors{Bart M. P. Jansen and Jan Arne Telle}
\EventNoEds{2}
\EventLongTitle{14th International Symposium on Parameterized and Exact Computation (IPEC 2019)}
\EventShortTitle{IPEC 2019}
\EventAcronym{IPEC}
\EventYear{2019}
\EventDate{September 11--13, 2019}
\EventLocation{Munich, Germany}
\EventLogo{}
\SeriesVolume{148}
\ArticleNo{20}
\usepackage{amssymb,amsmath,amsthm,amsfonts}
\usepackage[colorinlistoftodos]{todonotes}
\usepackage{xspace,xpunctuate}
\usepackage{xfrac}
\usepackage{thmtools}
\usepackage{hyperref}
\usepackage{mathtools}
\usepackage{tikz}
\usetikzlibrary{shapes,arrows,positioning}
\usepackage{enumerate}
\usepackage{subcaption}
\usepackage{ifpdf}
\usepackage{graphicx}
\ifpdf
\fi

\usepackage{framed}
\usepackage{mdframed}
\usepackage{float}

\newcommand{\N}{\ensuremath{\mathbf{N}}\xspace} % Natural numbers
\def\suchthat{such that\xspace}

\def\topgrad_#1{\widetilde \nabla\!_{#1}}

\newlength{\leftbarwidth}
\newlength{\leftbarsep}
\renewenvironment{leftbar}[1][blue]
{%
\def\FrameCommand
{%
{\hspace{-8pt} \color{#1} \vrule width 3pt}%
\hspace{0pt}%must no space.
\fboxsep=\FrameSep\colorbox{#1!5!}%
}%
\MakeFramed{\hsize\hsize\advance\hsize-\width\FrameRestore}%
}
{\endMakeFramed}
\usepackage{xparse}
\usepackage{cleveref}
\usepackage{comment}
\usepackage{graphicx}
\usepackage{tabularx}
\usetikzlibrary{calc}

\def\pp{path packing\xspace}
\def\PP{\textsc{Path Packing}\xspace}
\def\epc{exact path packing\xspace}
\def\EPC{\textsc{Exact Path Packing}\xspace}
\def\ppp{path-perfect packing\xspace}
\def\PPP{\textsc{Path-Perfect Packing}\xspace}

\newcommand{\mwnp}{\textsc{Multi-Way Number Partition}\xspace}

\def\L{{\cal L}}
\def\S{{\cal S}}
\newcommand{\np}{\textsc{NP}\xspace}
\newcommand{\xp}{\textsc{XP}\xspace}
\newcommand{\fpt}{\textsc{FPT}\xspace}

\def\PIP{\textsc{Length-$i$ Exact Edge Packing}\xspace}
\newcommand{\pthreep}{\textsc{Length-$3$ Exact Edge Packing}\xspace}
\newcommand{\ptwop}{\textsc{Length-$2$ Exact Edge Packing}\xspace}

\newcommand{\bcd}{\mbox{bcd}}

\newcommand{\short}[2]{#2}

\newtheorem{observation}[theorem]{Observation}

\begin{document}

\maketitle

\begin{abstract}
  We introduce and study the complexity of \PP.
  Given a graph \( G \) and a list of paths, the task is to embed the paths edge-disjoint in \( G \).
  This generalizes the well known \textsc{Hamiltonian-Path} problem.
  
  Since \textsc{Hamiltonian Path} is efficiently solvable for graphs of small treewidth, we study how this result translates to the much more general \PP.
  On the positive side, we give an \fpt-algorithm on trees for the number of paths as parameter.
  Further, we give an \xp-algorithm with the combined parameters maximal degree, number of connected components and number of nodes of degree at least three.
  Surprisingly the latter is an almost tight result by runtime and parameterization.
  We show an ETH lower bound almost matching our runtime.
  %Moreover, we show that if we drop one of these three parameters, the problem becomes para-\np-hard.
  Moreover, if two of the three values are constant and one is unbounded
  the problem becomes \np-hard.
  
  Further, we study restrictions to the given list of paths.
  On the positive side, we present an \fpt-algorithm parameterized by the sum of the lengths of the paths.
  Packing paths of length two is polynomial time solvable, while packing paths of length three is \np-hard.
  Finally, even the spacial case \EPC where the paths have to cover
  every edge in \( G \) exactly once is already \np-hard for two paths on 4-regular graphs.
\end{abstract}

\section{Introduction}

Packing, covering and partitioning are well researched fields in graph theory. 
In general, the task is to cover a given graph $G=(V,E)$ with or partition it into smaller substructures, or to pack given structures into the graph.
%se fields ask for a given graph $G=(V,E)$, where $V$ describes the set of vertices and $E$ the set of edges, if parts of the graph can be covered by or partitioned into smaller substructures or if given structures can packed into the graph. 
Besides that these terms are often used, they are not well defined throughout the literature. 
Thus, it is important to define problems in this field carefully and in detail. 

For example, the path partition problem is a well studied problem \cite{Arikati:1990:LAO:92325.92332, Bonuccelli:1979:MND:1710823.1710824,CHANGKUO96, Goodman:1975:AHC:321892.321897,MISRA197524,SRIKANT1993351,YAN1994317} which is also known as path cover problem. 
The task is to cover all vertices of a graph with vertex-disjoint paths. 
This is equivalent to partitioning the graph into vertex-disjoint paths. 
The smallest number of paths to achieve this is called the path partition number or path cover number. 
Observe that $G$ has a Hamiltonian path iff the path-partition number is one, thus the problem is NP-complete. 

An \np-complete variant of this problem is the $k$-path partition
problem~\cite{STEINER20032147, yan1997k,MONNOT2007677}. Here the task is to partition a graph
$G$ into paths, such that none of the path lengths exceeds
$k$. Observe that the 1-path-partition problem corresponds to finding a maximum matching. 

Another related problem is the recognition of path-perfect graphs \cite{cao2007solution,fink1981note,hamburger2005edge,straight1977partitions,zaks1977decomposition}, which we denote in this work as \PPP. 
Instead of partitioning the graph into vertex-disjoint paths, the complete edge set must be partitioned into edge-disjoint paths of ascending length, starting by one. 
This can also be understood as packing $k$ paths of length $1$ to $k$ into $G$ without using an edge twice or leaving one edge uncovered. 

This approach of packing smaller subgraphs into a given graph is also well researched \cite{yap1988packing}. 
For example, packing edge-disjoint trees into a clique is considered~\cite{straight1979packing}. 
Since packing edge-disjoint and vertex-disjoint triangles is \np-hard for planar graphs, the parameterized complexity is studied \cite{cornuejols1982packing}. 

We generalize the path-perfect graph problem and ask for a given graph $G$ and a list of $k$ paths $P = \{p_1,\ldots,p_k\}$ if they can be embedded into $G$ without using the same edge twice. 
Note that we define the length of a path equals its number of edges. 
This problem arises naturally when restricting the path partition problem
to edge-disjoint paths instead of vertex-disjoint paths.
We denote this problem as \emph{\PP}.
Let us formalize what we mean by embedding.
An embedding of a graph \( H \) into a graph \( G \) is an injective mapping \( f: V(H) \to V(G) \) such that for every original edge \( (u,v) \in E(H) \) also \( (f(u),f(v)) \in E(G) \).
An embedding of a list of graphs \( \mathcal{H} \) into \( G \) is an embedding of each graph \( H \) into \( G \).
Note, that we do not ask to embed the graphs pairwise
vertex-disjointly.
The embeddings we consider in this work are pairwise edge-disjoint embeddings of paths.

\medskip
\begin{tabularx}{0.95\textwidth}{ r X l }
	\multicolumn{2}{l}{\PP} \\
	Input: & A list of paths \( P = \{p_{1},\dots,p_{k} \} \) of length \(l_1,\ldots,l_k\). A graph \(G=(V,E)\). \\
	Question: & Is there an edge-disjoint embedding of $P$ into $G$?
\end{tabularx}
\medskip

%If the given paths have to be packed such that every edge in $G$ is covered by precisely one of the edge-disjoint paths we speak about the \emph{\EPC} problem.
The \emph{\EPC} problem additionally requires that every edge is covered \emph{exactly} once.

\medskip
	\begin{tabularx}{0.95\textwidth}{ r X l }
		\multicolumn{2}{l}{\EPC} \\
		Input: & A list of paths \( P = \{p_{1},\dots,p_{k} \} \) of length \(l_1,\ldots,l_k\). A graph \(G=(V,E)\). \\
		Question: & Is there an edge-disjoint embedding of $P$ into $G$ \suchthat each edge $e\in E$ is covered exactly once?
	\end{tabularx}
\medskip

\PP is clearly more general than \EPC, since one can reduce from one to the other by additionally requiring
the sum of the path lengths to be equal to the number of edges in the graph.
Most of our hardness results are for \EPC,
and therefore translate to \PP.
Our upper bounds are always regarding more the general \PP.

\paragraph*{Our Results}

The Hamiltonian path problem is a special case of \PP.
An even though the Hamiltonian path problem is tractable on graphs of bounded treewidth, 
\PP is already \np-complete on subdivided stars. 
Therefore, we focus on the parameterized complexity to classify this problem on a finer scale.
We will analyze the impact of various parameters.

In \Cref{sec:trees}, 
we analyze the parameterized complexity of our packing problems with respect to the number of paths (denoted by $k$).
On the one hand, we give an \fpt algorithm for \PP
that solves the problem in time $2^k n^{O(1)}$ on subcubic (i.e.\ degree at most three) forests (\cref{the:pp-dp-subcubic-forest}).
On the other hand, we show that \EPC on graphs with treewidth two is $W[1]$-hard and cannot be solved in time $f(k) n^{o(k/\log k)}$ under ETH (\cref{the:pc-w1hard-bounded-treewidth}).
%Thus, the problem is already hard on simple graph classes. 

In \Cref{sec:pathDependent} we introduce path dependent restrictions. 
We show that \EPC is \np-complete even for two paths on $4$-regular graphs (Theorem~\ref{the:ecp-npcomplete-two-paths}). 
length $i$ is easy for $i=2$ and \np-complete for $i=3$ (Theorem~\ref{the:pip-nphard}).
If we however parameterize by the summed length of all paths
\PP is in FPT~(Theorem~\ref{the:fptsummedpathlength}).

After parameterizing by the number of paths and their lengths, we further analyze graph dependent parameters in \Cref{sec:bcd}. 
We introduce the \emph{\bcd-number} of a graph, which is the maximum of the
number of components, the maximal degree, and the number of vertices with degree larger than two. 
We show that \PP can be solved in time $k!^k (n+k^2)^{O(k^2)}$, where $k$ is the \bcd-number~(\cref{the:pp-upper-bound-fpt}).
This is complemented by showing that the problem cannot be solved in $f(k) n^{o(k^2/\log k)}$ under ETH~(\cref{the:pp-lower-bound-fpt}).
We further show that all three bcd parameters are necessary:
If two values are constant and one is unbounded
the problem becomes \np-hard~(\cref{the:epc-nphard-forests}, \cref{col:NP-on-forests-of-paths}, \cref{the:ppc-npcomplete-caterpillars}).

Note that, one can embed paths $p_1,\dots,p_k$ as edge-disjoint subgraphs into a graph $G$ if and only if
one can embed these paths as vertex-disjoint induced subgraphs into the linegraph of $G$.
Therefore, our results yield new insights for the 
problem of covering a graph with a list of vertex-disjoint induced paths~\cite{le2003splitting}.
Especially, our hardness results for certain graph classes transfer
to hardness results on the linegraphs of these graph classes.
\short{Due to space limitations, we omit some proofs, and refer to the full version.}{}

\section{Preliminaries}
\label{section:preliminaries}

All graphs are simple (i.e. without multi-edges or self-loops). 
\short{The length of a path equals its number of edges.}{The length of a path $p$ is denoted by $|p|$ and equals its number of edges.}

\section{\PP on Forests} \label{sec:trees}

Our packing problem is a generalization of the Hamiltonian path problem and therefore \np-complete.
The Hamiltonian path problem is solvable in polynomial time if the treewidth of the input graph is bounded \cite{downey2012parameterized}. 
We show that (unlike Hamiltonian path) \EPC is \np-complete on trees.
This is done by reducing it to the following \np-complete partitioning problem.

\medskip
\begin{tabularx}{0.95\textwidth}{ r X p{0.5cm} }
	\multicolumn{2}{l}{\mwnp} \\
	Input: & A list of weights $w_1, \ldots w_n \in \N $ encoded in unary, and an integer \( k \in \N \). \\
	Question: & Is there a partition of \( w_1, \ldots w_n \) into \( k \) multi-sets \( S_1, \dots, S_k \) such that ${\sum_{w_i \in S_j} w_i = \frac{1}{k}\sum_{i=1}^{n} w_i}$, for every $1 \leq j \leq k$?
\end{tabularx}
\medskip
 
We reduce from \mwnp to prove that \EPC is \np-hard on very simple trees.

\begin{theorem}\label{the:epc-nphard-forests}
	\EPC is \np-complete on subdivided stars.%for star graphs.
\end{theorem}
\short{}{
\begin{proof}
	We construct an instance of the \epc problem from an instance of \mwnp problem as follows. 
	We create for each weight $w_j$ a path $p_j$ of length $w_j \cdot 2k$, with $1 \leq j \leq n$ and $n$ the number of weights. 
	So, the sum of the path lengths is $2k \cdot \sum_{j=1}^{n} w_j$. 
	To represent the sets in the \mwnp we construct $k$ paths of length $2 \cdot \sum_{j=1}^{n} w_j$ that share the center vertex $v$.
	Thus, we obtain a subdivided star where the center vertex $v$ is connected to $2k$ paths of length $\sum_{j=1}^{n} w_j$. 
	
	Figure \ref{StarReduction} shows the graph that results from an instance of the \mwnp problem with the weights $2,3,4,6,7,8$ that should be partitioned into three sets. 
	The overall sum of the weights is $30$ and therefore the graph in Figure \ref{StarReduction} has six paths of length $30$ connected to one vertex $v$.
	For each weight, we construct one path with the length equal to the weight multiplied with six. 
	Thus, the graph in Figure \ref{StarReduction} needs to be covered by six paths of length $12, 18, 24, 36, 42$ and $48$. 
	
	Next, we need to show that the constructed \epc instance has a solution if and only if the \mwnp instance is feasible. 
	Given a solution for the \epc problem on the constructed graph, we can compute a solution for the multi-way number problem as follows.
	A pair of two paths represent one set $S_i$ and the paths $p_j$ that cover the edges correspond to the weights $w_j$ that must be selected by the set $S_i$, for $1 \leq i \leq k$. 
	The pair is unique if a path $p_j$ covers edges of both paths, like the paths of length $36, 42$ and $48$ from the example in Figure \ref{StarReduction}. 
	Otherwise, there is an even number of edge covering paths that start or end at the center vertex $v$ and these paths can be combined arbitrarily to represents a set $S_i$. 
	There are $k$ many pairs of paths of the same length where all edges are covered. 
	Thus, the resulting sets refer to weights which sum up to $\frac{1}{k} \sum_{j=1}^{n} w_j$
	
	If we have a solution for the multi-way number problem we can choose a pair of paths adjacent to the center vertex $v$ to represent a set $S_i$. 
	The weights $w_j$ with $j \in S_i$ correspond to the paths $p_j$ that cover the chosen path pair. 
	The sum of the path lengths covers exactly one path pair. 
	
	\[ 
	\sum_{j\in S_i} |p_j| = 2k \cdot \sum_{j \in S_i} w_j  =  2 \cdot \sum_{j=1}^{n} w_j \text{.}
	\]
	\end{proof}
}

\short{}{
\begin{figure}
	\begin{center}
		\begin{tikzpicture}[x=0.6cm,y=0.8cm,>=latex,scale=0.6]
		\node[draw, circle, fill=lightgray, label={[label distance=0.09cm]180:$v$}] at (0,0) (v) {}; 
		%\filldraw (v) circle (2mm); 
		
		\node[draw, circle, fill=lightgray, scale=0.5] at (2,-3) (-32) {}; 
		\node[draw, circle, fill=lightgray, scale=0.5] at (2,-2) (-22) {}; 
		\node[draw, circle, fill=lightgray, scale=0.5] at (2,-1) (-12) {};
		\node[draw, circle, fill=lightgray, scale=0.5] at (2,3) (32) {}; 
		\node[draw, circle, fill=lightgray, scale=0.5] at (2,2) (22) {}; 
		\node[draw, circle, fill=lightgray, scale=0.5] at (2,1) (12) {}; 
		
		\foreach \row in {-3,-2,-1,1,2,3}
		{
			\draw (v) edge (\row2) {};
			\foreach \x in {3,4,...,31}{%12}{
				\pgfmathtruncatemacro\xi{\x-1}
				\node[draw, circle, fill=lightgray, scale=0.5] at (\x,\row) (\row\x) {};
				%\filldraw (\row\x) circle (1mm);
				\draw (\row\xi) edge (\row\x) {};
			}
		}
		\end{tikzpicture}
	\end{center}
	\caption{Subdivided Star construed from an instance of the \mwnp.}\label{StarReduction}
\end{figure}
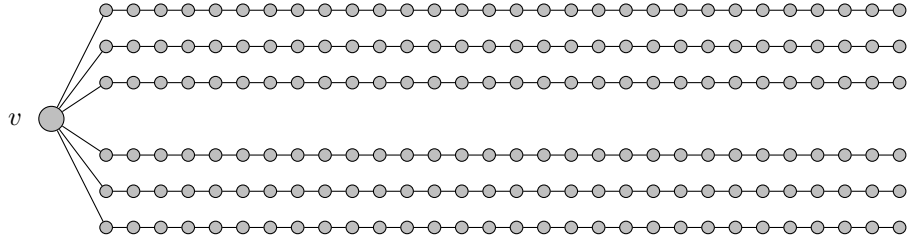
}

%So, in contrast to the Hamiltonian path problem, \PP is \np-complete on trees.
The previous reduction required a large number of paths.
%It appears like a large number of paths is important for the hardness on trees. 
Therefore, in the following, we analyze the \PP problem parameterized by the number of paths.

\paragraph*{Fast subset convolution}

We develop dynamic programming algorithms on subcubic
trees whose running time will be $O^*(2^k)$, where $k$ is the number of
paths that we want to pack.  First we develop a naive and not too
complicated algorithm with running time $O^*(3^k)$, whose longer
running time is due to some very simple operation that occurs when we
combine two dynamic programming tables.
Björklund, Husfeldt, Kaski, and Koivisto
introduced a technique called
\emph{fast subset convolution} that was used to speed up the computation of
Steiner trees with small integer weights~\cite{BHKK07} and also to speed
up some algorithms that do dynamic programming on tree
decompositions~\cite{RBR2009}.  We can use this technique to our advantage
to significantly speed up the path packing algorithm on trees.  The
result that we will be using is:
\begin{proposition}{\rm\cite{BHKK07}}
\label{prop:fast-subset-convolution}
Let $N$ be a set of $n$ natural numbers and $f,g\colon\N\to\N$ two functions.
Then we can compute $(f*g)(S)$ for all $S\subseteq N$ in time
$O(2^nn^2)$ if $f$ and $g$ can be evaluated in constant time and where
$
(f*g)(S)=\sum_{T\subset S} f(T)g(S-T).
$
Here $f(S)=\sum_{i\in S}f(i)$.
\end{proposition}

\begin{figure}[t]
\centerline{\includegraphics{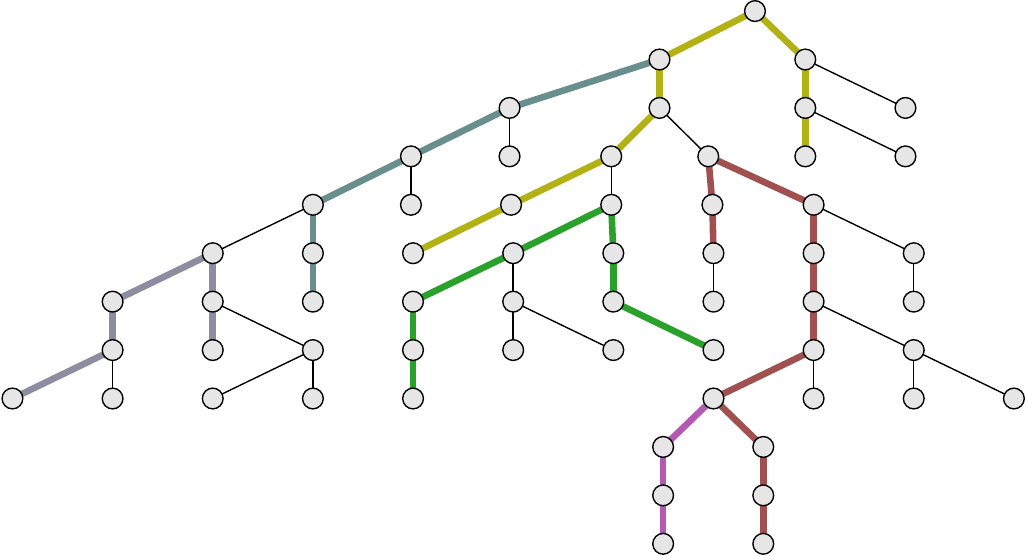}}
\caption{Packing paths of lengths 10, 8, 7, 5, 5, 3 into a subcubic
tree.  Although the packing looks very loose there is no solution
if we replace 3 by~4.}
\label{fig:pp-example1}
\end{figure}

\begin{theorem}\label{the:pp-dp-subcubic-forest}
We can solve \PP for $k$ paths in time $O^*(2^k)$ on
subcubic forests.
\end{theorem}

\begin{proof}
Let us assume that the graph is a subcubic tree~$T$, but the proof
easily generalizes to subcubic forests.
Let $l_1,\ldots,l_k$ be
length of the paths that we want to pack into~$T$.  We can further
assume that $T$ is a rooted tree by designating an arbitrary vertex as
its root.  If $v$ is a vertex of $T$ then let $T(v)$ be the subtree
rooted at~$v$.

We solve the \pp problem by dynamic programming computing a table for
each vertex in a bottom-up order.  Such a table is a mapping $L\colon
V\times 2^{[k]}\to [n]\cup \{-\infty\}$.  The size of this table
is $O(2^kn)$.  We interpret the content of the table as follows:

$L(v,P)=r$ with $r\geq0$ means that it is possible
to pack all paths
with indices in $P$ (in short all $P$-paths)
into the subtree $T(v)$ and additionally a path of length~$r$
that ends in~$v$.

The special case $L(v,P)=0$ means that we can pack all $P$-paths
into $T(v)$, but however we pack them there is no space left to pack
another path that ends in~$v$.

%\begin{figure}
%\centerline{\includegraphics[width=0.4\hsize]{figures/tw1.jpg}}
%\caption{The tree has root~$v$.  There are four path that we want to
%pack.  We have $L(v,\{1,2,3\},4)=2$ because it is possible to pack
%paths 1,2,3 into the tree such path~4 can be packed, too, when
%sticking out at the top.  It has to stick out at least by two edges.}
%\label{fig:tw1}
%\end{figure}

If it is not possible to pack all $P$-paths into $T(v)$ at all then
let $L(v,P)=-\infty$.

It is quite clear that having computed all tables
enables us to find out whether $(T,P)$ is a yes-instance of the
\pp-problem.  Simply check whether $L(r, \{1,\ldots,k\})\neq-\infty$.

To compute the tables for all $v$ we distinguish three cases how to
compose trees into bigger trees:
{\bf 1} $v$ is a leaf,
{\bf 2} $v$ has one child,
{\bf 3} $v$ has at least two children.

{\bf Leaf. }%\paragraph*{Leaf.}
If $v$ is a leaf then $L(v,\emptyset)=0$ and $L(v,P)=-\infty$ if
$P\neq\emptyset$ because we cannot pack any path into an empty tree
(that has no edges).

{\bf One child. } %\paragraph*{One child.}
\begin{figure}
	\centerline{\includegraphics{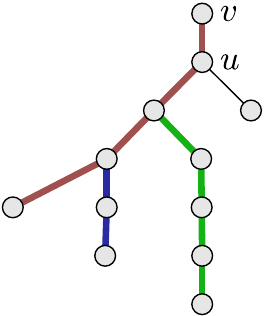}\hfil
		\includegraphics{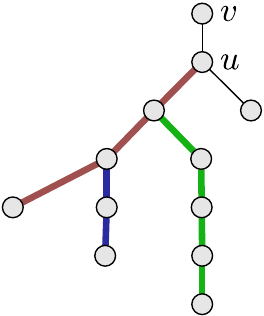}}
	\caption{Left side: Packing paths of lengths $l_1=4$, $l_2=4$, $l_3=2$ into
		$T(v)$.  $L(u,\{1,2,3\})=-\infty$, but $L(u,\{1,2,3\}-\{1\})=l_1-1$,
		so $L(v,\{1,2,3\})=0$.\\
		Right side:  Now $l_1=4$, $l_2=3$, $l_3=2$.  $L(u,\{1,2,3\})=1$, so
		$L(v,\{1,2,3\})=1+1=2$.  An additional path of length~2 can be packed
		into $T(v)$, because an additional path of length~1 can be packed
		into~$T(u)$.}
	\label{fig:tw2}
\end{figure}
If $v$ has one child $u$ then it is also quite easy to compute
$L(v,P)$:  If $L(u,P)=r$ with $r\geq0$, then clearly $L(v,P)=r+1$.
The right hand side of Figure~\ref{fig:tw2} shows an example.
The more
complicated possibility is $L(u,P)=-\infty$, which means that it is
completely impossible to pack all $P$-paths into $T(u)$.  It might
become possible to pack all $P$-paths into $T(v)$ by using the
additional edge~$uv$.  If this is possible, then one path, say the $i$th
one with length~$l_i$, uses the edge~$uv$.  Then all paths in $P-\{i\}$
are packed into $T(u)$ and one additional path of length $l_i-1$ that
ends in~$u$.  We can check this by verifying that $L(u,P-\{i\})\geq
l_i-1$ for some~$1\leq i\leq k$ (actually, $L(u,P-\{i\})\geq l_i$ is
impossible, because then all $P$-paths could be packed into $T(u)$ and
$L(u,P)\neq-\infty$).  If we find such an $i$, then all $P$-paths can be
packed into $T(v)$, but only by using the edge~$uv$.  This means
that no other path can be packed into $T(v)$ that ends in~$v$ and
therefore $L(v,P)=0$.  See the left hand side of Figure~\ref{fig:tw2} for
an example.
\[
L(v,P)=
\begin{cases}
  L(u,P)+1&\text{if } L(u,P)\geq0\\[5pt]
  0       &\text{if $L(u,P)=-\infty$ and $L(u,P-\{i\})=l_i-1$ for some
           $i\in P$}\\
  -\infty &\text{otherwise}
\end{cases}
\]

{\bf Two children. } %\paragraph*{Two children.}
Finally, we assume that $v$ has exactly two children $u_1,u_2$.
In that case we can construct for each of them a new tree by
attaching new roots $v_1,v_2$ to $T(u_1)$ and $T(u_2)$
and computing
the $L$-tables for both of them.  To compute the table of $v$ it is
sufficient to compute a table for a tree that we get by glueing two
trees together by identifying their roots.  We just have to 
glue $v_1$ to~$v_2$.

So we can assume that we have two trees with roots $v_1$ and $v_2$ and
a tree with root $v$ that we get by identifying $v_1$ and $v_2$ and
renaming it to~$v$.  This is often called a join operation.
We have the tables for $v_1$ and $v_2$ and want
to compute the table for~$v$.

Clearly, $L(v,P)=r$ with $r>0$  
iff some of the $P$-paths can be packed into $T(v_1)$ and
the others into $T(v_2)$ and the additional path with length~$r$
that ends in $v$ can be packed into $T(v_1)$ or $T(v_2)$.
The additional path of length $r$ that ends in $v$ prevents any
$P$-path from being packed partially into $T(v_1)$ and $T(v_2)$.
That is the case iff there is a bipartition of $P$ into $P_1$ and $P_2$
such that $L(v_1,P_1), L(v_2,P_2)\geq 0$ and
$\max\{L(v_1,P_1),L(v_2,P_2)\}=r$.
There are
$2^{|P|}$ many subsets of~$P$.  To check all
bipartitions for all these subsets $P\subseteq [k]$ means looking at
$
\sum_{i=0}^k {k\choose i}2^i=3^k
$
many cases.
Using fast subset convolution lets us speed up the
computation to $2^k k^{O(1)}$ steps:
Let 
\[
f_i(S)=
\begin{cases}
1 & \text{if } L(v_i,S)\geq0\\
0 & \text{otherwise}
\end{cases}
\qquad
g_i(S)=
\begin{cases}
1 & \text{if } L(v_i,S)\geq r\\
0 & \text{otherwise.}
\end{cases}
\]
Then $L(v,P)\geq r$ iff $(f_1*g_2)(P)+(g_1*f_2)(P)\geq1$.  

\begin{figure}
	\centerline{\includegraphics{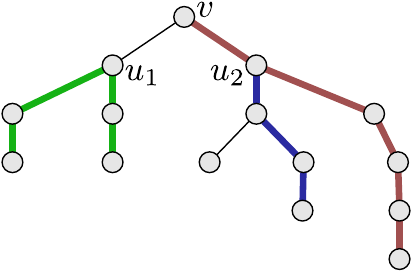}\hfil
		\includegraphics{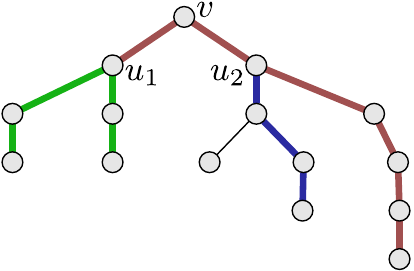}}
	\caption{Left side: Packing paths of lengths $l_1=5$, $l_2=4$, $l_3=3$ into
		$T(v)$.  $L(v_1,\{2\})=1$ and $L(v_2,\{1,3\})=0$ imply
		$L(v,\{1,2,3\})\geq1$.\\
		Right side:  Now $l_1=6$, $l_2=4$, $l_3=3$.
		$L(v_1,\{2\})+L(v_2,\{3\})=1+5=l_1$ implies 
		$L(v,\{1,2,3\})\geq0$.  This time we partition $\{1,2,3\}$ into three
		parts.  One goes to the left, one to the right, and one path uses both
		subtrees.}
	\label{fig:tw4}
\end{figure}

The situation is different if $L(v,P)=0$.
In that case both edges $u_1v$ and $u_2v$ have to be used when packing
all $P$-paths into $T(v)$ because otherwise at least a path of length
one that ends in $v$ could additionally be packed into~$T(v)$.

In such a packing one path, say the $i$th one with length~$l_i$,
uses $u_1v$ and~$u_2v$.  That is possible iff there is a bipartition of
$P-\{i\}$ into $P_1$ and $P_2$ such that $L(v_1,P_1)+L(v_2,P_2)\geq
l_i$.  Again, by using fast subset convolution we can check this in
$2^kk^{O(1)}$ steps.
\end{proof}

The above proof does not work any more if we glue together two trees whose
roots have degree higher than one.  For general trees the dynamic
programming is much more complicated and we will need more complicated tables.

Let $T(v)$ be a rooted tree with root $v$ and no restrictions on the
degree of vertices (and thus on the number of children).  Let us again
fix length $l_1,\ldots,l_k\in\N$ of paths that we are going to pack
into a tree.  We are going to identify a set of paths by a set
$P\subseteq[k]$.  We speak of $P$-paths as the paths with length $l_i$
for every $i\in P$.

\begin{definition}\rm
Let us fix $l_1,\ldots,l_k\in\N$, $P\subseteq[k]$, and $T$ be a rooted
tree.  $T(v)$ is the subtree of $T$ with root~$v$.
\begin{enumerate}
\item Let $M,M'\subseteq\N$ be two multisets.  We say that
$M'\succcurlyeq M$ if we can construct $M'$ from $M$ by adding numbers 
and increasing numbers that are already in~$M'$.

Let $M'\succ M$ iff $M'\neq M$ and $M'\succcurlyeq M$.

Example: $\{3,3,5,5,7\}\succcurlyeq\{2,3,4,6\}$,
but $\{3,3,5,5,7\}\not\succcurlyeq\{2,3,4,8\}$.
\item Let $\S$ be a set of multisets of natural numbers.  Then
\[
K(\S)= \{\,M\in\S\mid \text{ there is no }M'\in\S\text{ with }
M'\succ M\,\}.
\]

\item Then we define $\L(v,P)$ as a set of multisets of natural numbers as
follows:

Let $M\subseteq\N$ be a multiset of natural numbers.  Then
$M\in\L(v,P)$ iff it is possible to pack all $P$-paths into $T(v)$
such that we can pack additionally all non-empty paths into $T(v)$
that start at $v$ and have lengths given in~$M$ and if there is no $M'\in\L(v,P)$ with $M'\succcurlyeq M$.

Particularly, $\L(v,P)=\emptyset$ iff it is impossible to pack all
$P$-paths into $T(v)$ and $\L(v,P)=\{\emptyset\}$ iff it is possible
to pack all $P$-paths into $T(v)$, but there is no possibility to
additionally pack a non-empty path that starts at~$v$.

\item If $M\subseteq\N$ then $\max_q(M)$ is the multiset that consists
of the $q$ biggest elements in $M$ or of all of them if $M$ contains
less than $q$ numbers, e.g., $\max_3(\{5,5,4,4,3,2,1\})=\{5,5,4\}$.

%We do not need these two definitions anymore:
%\item If $\S$ is a set of multisets, then $R_q(\S)=\{\,\max_q(M)\mid
%M\in\S\,\}$.
%
%\item Let $\S_1,\S_2$ be two sets of multisets of natural numbers.
%Then $\S_1\bowtie\S_2=\{\,M_1\cup M_2\mid M_1\in\S_1\text{ and }
%M_2\in\S_2\,\}$.
\end{enumerate}
\end{definition}

\begin{figure}
	\centerline{
	\includegraphics[width=\textwidth]{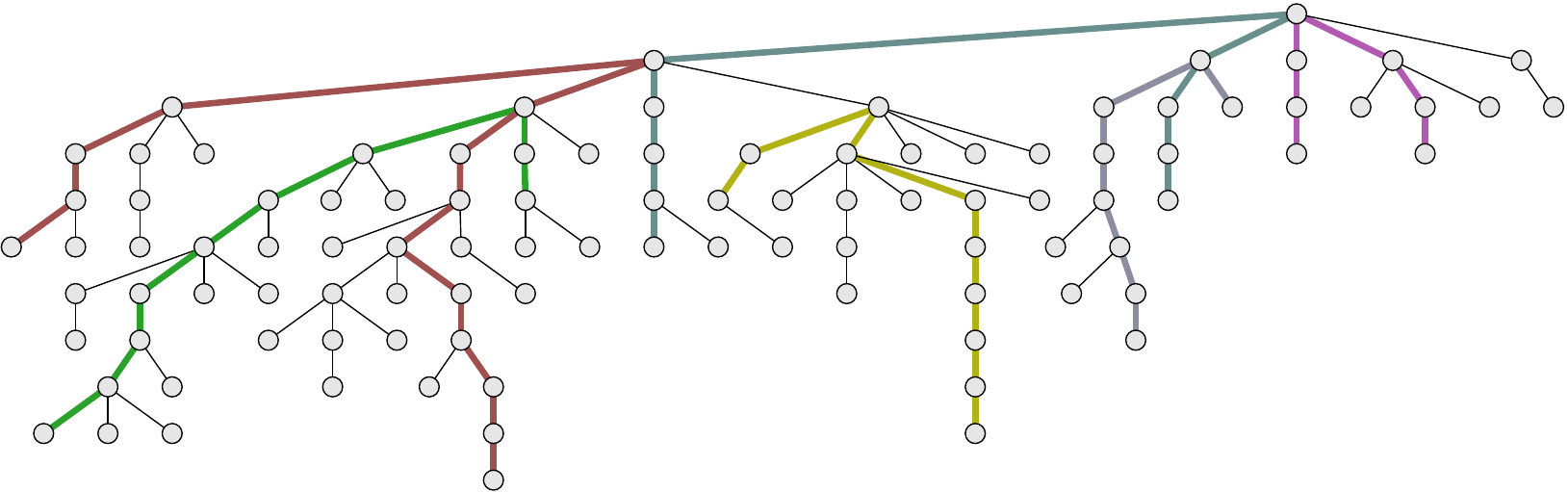}}
	\caption{In this tree there are nodes with more than two children and
paths can ``cross.''  We pack paths with lengths $13,9,9,9,7,6$.  There
is no solution if we replace 6 by~7.}
\label{fig:tw6}
\end{figure}

In the following let $l_1,\ldots,l_k$ be fixed.

%\short{\newpage}{}  This was crazy and botched up the whole page layout...
\begin{lemma}
Let $T(v)$ be a rooted tree with root $v$ such that $v$ has one
child~$u$.
\begin{enumerate}
\item
Assume that $\L(u,P)=\{M_1,\ldots,M_m\}$ with $m\geq1$.
Define $\L_{\max}(u,P)=\max(M_1\cup\cdots\cup M_m)$ (where $\max\emptyset=0$).

Then $\L(v,P)= \{\{\L_{\max}(u,P)+1\}\}$.
\item
Assume that $\L(u,P)=\emptyset$ and there is an $i\in\{1,\ldots,k\}$ with
$\L_{\max}(u,P-\{i\})=l_i-1$.

Then $\L(v,P) = \{\emptyset\}$.
\item
Otherwise $\L(v,P) = \emptyset$.
\end{enumerate}
\end{lemma}

\begin{proof}
We have to consider exactly two cases.  The first case is that it
is possible to
pack all $P$-paths into $T(u)$.  If this is the case, then an
additional path of length $r+1$ can be packed into $T(v)$ starting at
$v$ iff an additional path of length $r$ can be packed into $T(u)$
starting at~$u$.  The latter is the case iff $\L_{\max}(u,P)=r$.

The second case is that it is impossible to pack all $P$-paths into
$T(u)$ alone.  It might still be possible to pack them into $T(v)$,
but only if the edge $uv$ is used.  This means that there is only
space for an additional path of length zero that starts at~$v$.

In fact, exactly one path, say the $i$th one, uses the edge $uv$.
This is possible iff we can pack all $(P-\{i\})$-paths into $T(u)$ and
being able to additionally pack a path of length at least $l_i-1$ into
$T(u)$ starting at~$u$.  Actually, this path cannot be longer than
$l_i-1$ because then we would be able to pack all $P$-paths, which is
a contradiction.
\end{proof}

\begin{lemma}\label{lem:dpgc2}
Let $T(v_1)$ and $T(v_2)$ be two rooted trees with no common vertices,
such that $v_2$ has exactly one child.
Let $T(v)$ be the tree that we get by identifying $v_1$ with $v_2$ and
renaming it to~$v$. Then $\L(v,P)=K(\L_1\cup\L_2)$ where
\begin{align*}
\L_1&=\bigcup_{P_1\subseteq P\atop P_2=P-P_1}
      \bigcup_{M_1\in \L(v_1,P_1)\atop M_2\in\L(v_2,P_2)}
      \bigl\{M_1\cup M_2\bigr\}\\
\L_2&=\bigcup_{P_1\subseteq P\atop P_2=P-P_1}
      \bigcup_{i\in P}
      \bigcup_{M_1\in \L(v_1,P_1-\{i\})\atop M_2\in\L(v_2,P_2-\{i\})}
      \bigcup_{{r_1\in M_1\atop r_2\in M_2}\atop r_1+r_2\geq l_i}
        \bigl\{(M_1-\{r_1\})\cup(M_2-\{r_2\})\bigr\}
\end{align*}
\end{lemma}

\begin{proof}
``$\L(v,P)\supseteq K(\L_1\cup\L_2)$'':
If $M\in \L_1\cup\L_2$ then $M\in\L_1$ or $M\in\L_2$.  Let us first
consider the case $M\in\L_1$.  By the definition of $\L_1$ there
are $P_1\subseteq P$, $P_2=P-P_1$, $M_1\in\L(v_1,P_1)$, and
$M_2\in\L(v_2,P_2)$ such that $M=M_1\cup M_2$.  By induction we know
that $P_1$ can be packed into $T(v_1)$ as well as additional paths of
lengths $M_1$ starting at $v_1$.  The same holds for $P_2$, $v_2$, and
$M_2$.  Using this packing we actually packed $P$ into $T(v)$ and
additional paths of lengths $M_1\cup M_2=M$ starting at~$v$.
By definition then $M\in\L(v,P)$.

The other possibility is $M\in\L_2$, which is a bit more complicated.
If $M\in\L_2$, then $M=(M_1-\{r_1\})\cup(M_2-\{r_2\})$, where
$r_1\in M_1$, $r_2\in M_2$, $r_1+r_2\geq l_i$,
$M_1\in\L(v_1,P_1-\{i\})$,
$M_2\in\L(v_2,P_2-\{i\})$,
$P_1\subseteq P$, $P_2=P-P_1$, and $i\in P$.

We have to show that it is possible to pack $P$ into $T(v)$ and
additionally paths with lengths from $M$ starting at~$v$.
By induction we know that we can pack all $(P_1-\{i\})$-paths into
$T(v_1)$ and all $(P_2-\{i\})$-paths into~$T(v_2)$.  Simultaneously,
we can pack additional paths with lengths from $M_1$ into $T(v_1)$
starting at~$v_1$ and paths with lengths from $M_2$ into $T(v_2)$
starting at~$v_2$.  Hence, we can pack paths with lengths $M=
(M_1-\{r_1\})\cup(M_2-\{r_2\})$ into $T(v)$ leaving space for a path
of length $r_1$ in $T(v_1)$ and a path of length $r_2$ in $T(v_2)$.
We can combine these two paths into one path of length $r_1+r_2\geq
l_i$ and pack one additional path of length~$l_i$ into $T(v)$.
Altogether we packed $P_1$, $P_2$, $\{i\}$ and therefore all $P$-paths
into $T(v)$.

``$\L(v,P)\subseteq K(\L_1\cup\L_2)$'':
Let $M\in\L(v,P)$.  Then $P$ can be packed into~$T(v)$.  There are two
possibilities:

1. No path corresponding to $i\in P$ lies partially in $T(v_1)$ and
partially in $T(v_2)$.  Then we can split $P=P_1\cup P_2$ such that
$P_1$-paths are packed into $T(v_1)$ and $P_2$-paths into $T(v_2)$.
The additional path with lengths from $M$ are also packed into
$T(v_1)$ and $T(v_2)$.  Let us say $M=M_1\cup M_2$, where $M_1$ is in
$T(v_1)$ and $M_2$ in $T(v_2)$.  Then it is easy to see that
$M\in\L_1$.

2. There is an $i\in P$ such that all $(P-\{i\})$-paths are packed
into $T(v_1)$ and $T(v_2)$, but exactly one path with length $l_i$
is packed into $T(v)$ using edges from both $T(v_1)$ and~$T(v_2)$.
Note that there can be at most one such path because $v_2$ has only
one child in~$T(v_2)$.  Then all additional paths with lengths in $M$
that start at $v$ have to be packed into $T(v_1)$ alone because the
edge in $T(v_2)$ is not available any more.  Let $r_1$ be the length
of the part of the bridging path of length~$l_i$ that lies in $T(v_1)$
and $r_2$ the length of the part in~$T(v_2)$.  Clearly, $r_1+r_2=l_i$.
With all these facts we can again easily verify that $M\in\L_2$.
\end{proof}

%\begin{lemma}
%Let $T(v_1)$ and $T(v_2)$ two rooted trees with no common vertices,
%such that $v_2$ has exactly one child.
%Let $T(v)$ be the tree that we get by identifying $v_1$ with $v_2$ and
%renaming it to~$v$.  Then 
%$R_q(\L(v,P))=R_q\bigl(K({\cal R}_1\cup {\cal R}_2)\bigr)$, where
%\begin{align*}
%{\cal R}_1&=\bigcup_{P_1\subseteq P\atop P_2=P-P_1}
%      \bigcup_{M_1\in R_q(\L(v_1,P_1))\atop M_2\in R_q(\L(v_2,P_2))}
%      \bigl\{M_1\cup M_2\bigr\}\\
%{\cal R}_2&=\bigcup_{P_1\subseteq P\atop P_2=P-P_1}
%      \bigcup_{i\in P}
%      \bigcup_{M_1\in R_q(\L(v_1,P_1-\{i\}))\atop M_2\in R_q(\L(v_2,P_2-\{i\}))}
%      \bigcup_{{r_1\in M_1\atop r_2\in M_2}\atop r_1+r_2\geq l_i}
%        \bigl\{(M_1-\{r_1\})\cup(M_2-\{r_2\})\bigr\}.
%\end{align*}
%\end{lemma}
%
%\begin{proof}
%\dots
%\end{proof}

The following lemma shows that the size of the tables is bounded by a
function in $k$ and the maximal degree.  The estimate is quite
pessimistic, but we are not trying to optimize the runtime of the
dynamic programming algorithm at the moment and are content with
proving fixed parameter tractability.

\begin{lemma}\label{lem:Lsmall}
Let $T(v)$ be a rooted tree and assume that vertex $v$ has $d$ children. 
Then
\[
{|\L(v,P)|\leq d2^{kd}}.
\]
\end{lemma}

\begin{proof}
If $v$ has only one child, then $|\L(v,P)|=1$ and the statement is
true.  Assume next that $T(v)$ has $d$ children.  Each subtree can
receive at most $2^k$ different sets of packed paths yielding at most
$2^k$ different length of the longest path that can be additionally
packed.  Therefore a set $M\in\L(v,P)$ can have size at most $d$ and
contain up to $d$ numbers each chosen from a set of size at most
$2^k$.  In total that are at most $d2^{kd}$ possibilities for a
set~$M$.
\end{proof}

\begin{theorem}\label{thm:3P}
Let $T$ be a rooted tree and $P$ a multiset of paths.  In
polynomial time a rooted tree $T'$ can be computed that has the
following properties:\\
(1)~$P$ can be packed into $T$ iff it can be packed into~$T'$,\\
(2)~each node in $T'$ has at most $3|P|$ children, and (3)~$T'$ is a subtree of~$T$.
\end{theorem}

\begin{proof}
Let $l_1(u)$ be the length of the longest path in $T(u)$ that starts
in~$u$ and $l_2(u)$ be the length of the longest path in $T(u)$.
Assume that $P$ can be packed into~$T$ and $v$ be an arbitrary vertex
in~$T$.  Let us fix an edge-disjoint packing of~$P$. 

Let $v$ be an arbitrary node in~$T$ and $v_1,\ldots,v_m$ the children
of~$v$.  Let us further assume that $v_1,\ldots,v_{3|P|}$ contain
the $|P|$ children with biggest $l_1(v_i)$ and $2|P|$ children with
biggest $l_2(v_i)$.  Ties can be arbitrarily ordered.

If $m\leq 3|P|$ we do nothing.  Otherwise assume that $P$ is packed into
$T$ and some path $p\in P$ uses $T(v_i)$ with $i>3|P|$.  There are two
possibilities:

(i)~$p$ contains~$v_i$.  Then $p$ is possibly packed
partially inside $T(v_i)$ and partially outside.  Let $p'$ be the part
of $p$ inside~$T(v_i)$.  Clearly, $p'$ starts at $v_i$. 
By the pigeonhole principle there must be some $T(v_k)$ that has not
been used in the packing of~$P$, $l_1(v_k)\geq l_1(v_i)$,
and $k\leq3|P|$.  Then we can repack
$p$ such that it uses $T(v_k)$ instead of~$T(v_i)$.

(ii)~$p$ does not contain~$v_i$ and is therefore completely packed
into $T(v_i)$.  Again by the pigeonhole principle we can find an
appropriate $T(v_k)$ with $k\leq 3|P|$ and $l_2(v_k)\geq l_2(v_i)$.
We can repack $p$ from $T(v_i)$ into~$T(v_k)$.

Repeated repacking in these two ways leads to a packing that uses
only the subtrees $T(v_1),\ldots,T(v_{3|P|})$.
We can therefore remove all other subtrees without changing a
yes-instance into a no-instance.  Applying this pruning to all
vertices in~$T$ leads to a new tree $T'$ that has all properties
stated in the theorem.  It is also clear that $T'$ can be computed in
polynomial time as it is easy to find longest paths in trees.
\end{proof}

%\begin{lemma}
%Let $T(v)$ be a rooted tree.
%Then $|R_q(\L(v,P))|\leq \ldots$.
%\end{lemma}
%
%\begin{proof}
%Let us assume that $M\in\L(v,P)$ and that $m_1,\ldots,m_q$ are the
%$q$ biggest elements in~$M$.  Furthermore let $v$ have $d$ children
%$u_1,\ldots, u_d$ and let $T(v_1),\ldots,T(v_d)$ be trees that we get
%by attaching a new root $v_i$ to $u_i$ in $T(u_i)$.  If it is indeed
%possible to pack all $P$-paths into $T(v)$ then some of them could
%start at~$v$ and some of them could have $v$ as an inner node.  This
%results in some of the edges $u_iv$ being used in the packing of the
%$P$-paths, but only $2|P|$ of such edges can used in total.  Let $I$
%contain $i$ iff $u_iv$ is used in a particular fixed packing of~$P$.
%
%\dots
%\end{proof}

Combining the above results (with the base case for a leaf $v$:
$\L(v,P)=\{\{0\}\}$ if $P$ contains only empty paths and 
$\L(v,P)=\emptyset$ otherwise) we can prove the following:

\begin{theorem}\label{the:pp-fpt-tree}
\PP into forests parameterized by the number of paths is in \fpt.
\end{theorem}

\begin{proof}
Given a tree $T$ compute a rooted tree $T'$ where each node has at
most $3k$ children and every $P$ (with $|P|=k$) can be packed into
$T$ iff it can be packed into~$T'$ (Theorem~\ref{thm:3P}).
Then use dynamic programming 
to find out whether the paths can be packed into $T'$.  By
Lemma~\ref{lem:Lsmall} and~\ref{lem:dpgc2} this only takes
time $f(k){|T|^{O(1)}}$ for some function~$f$.
\end{proof}

\paragraph*{Lower bound}

%We complement the positive result for \PP parameterized by the number of paths
While \PP on graphs with treewidth one is in FPT when parameterized by the number of paths,
we now show that the problem becomes hard on graphs with treewidth two.
As an intermediate step, we reduce from \textsc{Unary Bin Packing}~\cite{jansen2013bin} to show hardness of \mwnp.
This then leads to hardness results for \EPC. 
Remember that for \mwnp the numbers are unary encoded.

\begin{lemma}\label{lem:multiwayhard}
	\mwnp parameterized by the number of sets $k$ is $W[1]$-hard. 
    Moreover, unless ETH fails there is no algorithm that solves the problem in $f(k)N^{o(k/\log k)}$ time for some function $f$ where $N$ is the input size.
\end{lemma}
\short{}{
\begin{proof}
	We give an \fpt-reduction from \textsc{Unary Bin Packing} with the number of bins \( k \) as parameter.
	The input to this problem is a a list of weights $w_1,\dots,w_n$ encoded in unary and $k \in \N $ bins, each with capacity $b$. 
    The task is to decide if the weights can be packed into the $k$ bins.
	This problem is \( W[1] \)-hard, and cannot be solved in time $f(k)N^{o(k/\log k)}$ where \( N \) is the input size unless ETH fails \cite{jansen2013bin}.
	
	Given an instance of \textsc{Unary Bin Packing}, we construct an instance of \mwnp as follows.
    If the sum of weights \( \sum_{i=1}^n w_i \) exceeds the total capacity, \( b k \), return a trivial no-instance.
    Otherwise, let \(bk -  \sum_{i=1}^n w_i = d \geq 0 \) be the remaining space after all items are packed (assuming this is possible).
    If \( d \geq \sum_{i=1}^n w_i \), note that a folklore 2-approximation of bin packing guarantees a solution.
    In this case, we return a trivial yes-instance.

    Otherwise, we start constructing the following \mwnp instance:
    Additionally to the old weights, we add add \( d \) many `dummy' elements with weight one.
	The number of partitions \( k \) is the number of bins.
    Doing so increases the input size at most by a factor of two, since we add at most \( \sum_{i=1}^n w_i \) weights in unary.
	
	We claim that the constructed \mwnp-instance has a solution if and only if the original \textsc{Unary Bin Packing} has a solution.
	A solution of \mwnp partitions the weights and `dummy' weights into multi-sets of sum \( bk/k = b \).
	Thus there is a bin packing of the weights (without `dummy' weights) into \( b \) bins of size \( k \).
	For the other direction, a solution to the bin packing gives a partition into \( k \) sets, each of sum \( \leq b \), and total sum \( bk-d \).
	Thus, filling the bins with the \( d \) `dummy' weights yields a solution for the \mwnp-instance.
	
	The run time is polynomial in the input size \( N \) and the parameter \( k \) remains equal.
	Since \textsc{Unary Bin Packing} is \( W[1] \)-hard and has an ETH based lower bound of $f(k)N^{o(k/\log k)}$, these hardness results translate to \mwnp.
\end{proof}
}

%We showed earlier then \EPC is $W[1]$-hard if the number of paths is a fixed constant
%For arbitrary but fixed $k$, the \EPC problem is $W[1]$-hard, even on graphs with treewidth two. 
%Because we reduce from \mwnp the lower bound result is transferred. 
%Thus, for graphs with treewidth one, we are able to solve the problem if we parametrize it by the number of paths and show hardness for treewidth two. 

\begin{theorem}\label{the:pc-w1hard-bounded-treewidth}
	The \EPC problem parameterized by the number of paths on graphs with treewidth two is $W[1]$-hard.
    Moreover, unless ETH fails there is no algorithm that solves the
    problem in $f(k) \; n^{o(k/\log k)}$ time for some function $f$ where $k$ is the
    number of paths and $n$ the number of vertices in the input graph.
\end{theorem}
\short{}{
\begin{proof}
	We give an \fpt reduction from \mwnp, parameterized by the number of sets. %unary bin packing. 
	%Given a set of $n$ integer items, with sizes encoded in unary, and two integers $b$ and $k$, the task is to decide if the items can be packed into $k$ bins of capacity $b$. 
	%The problem is known to be $W[1]$-hard, parameterized by the number of bins $k$ \cite{jansen2013bin}. 
	%We construct an instance of the \EPC problem with $k$ paths from an instance of the \mwnp problem with $k$ sets.
    %The constructed graph $G$ will be series-parallel and its number of vertices
    %will be linear in the unary encoding of the weights.
    %By \cref{lem:multiwayhard}, this is sufficient.
    Assume we want to partition weights $w_1,\dots,w_b$ into $k$ sets.
    By \cref{lem:multiwayhard} it is sufficient to reduce such an \mwnp-instance
    to an instance of the \EPC problem with $k$ paths
    and a series-parallel graph $G$ 
    whose number of vertices will be linear in $N = \sum_{i=1}^b w_i$.
    We proceed as follows. 

    %For every weight $w_i$, we want to embedd a path $l_i$ of length $w_i+2$.

    %We construct a series-parallel graph $G$ that contains $k-1$ edge-disjoint paths $(v_0,u^i_1,v_1,u_2^i,v_2, \ldots ,u^i_{b}, v_b)$ of length $2b$, for $1 \leq i \leq k-1$. 
    At first, we construct a graph $G$:
    We add vertices $v_0,\dots,v_b$.
	Additionally, every vertex pair $v_{j-1}, v_j$ with $1 \leq j \leq b$
    we add $k-1$ paths of length $2$
    and one path of length $w_j+2$ connecting them.
    At last, we add a set $S$ of $k$ private neighbors to $v_0$
    and a set $T$ of $k$ private neighbors to $v_b$.
	Figure \ref{seriesParallelReduction} shows the graph constructed from a \mwnp-instance with weights $3,10,8,2,6,4,7,5,\dots$ and $k=4$. 
	We further specify $k$ paths of length $2(b+1) + N/k$ each. %$b+2(n+1)$. 
    Obviously, $G$ is series-parallel and its number of vertices is bounded by $O(N)$.
    It remains to show that the $k$ paths of length $2(b+1)+N/k$ can be embedded into $G$ if and only
    if $w_1,\dots,w_b$ can be partitioned into $k$ sets
    of size $N/k$.

    Since every feasible packing has to cover every edge exactly once,
    the $k$ paths all have to start in $S$, visit $v_0,\dots,v_b$ in this order and end in $T$.
    Assume we have embedded paths $p_1,\dots,p_k$ of arbitrary length
    into $G$ such that they start in $S$, visit $v_0,\dots,v_b$ in this order and end in $T$.
	Between vertices $v_{j-1}$ to $v_j$, for $1 \leq j \leq b$, these paths have options:
	Exactly $k-1$ paths have to take the direct path of length $2$
    and one path has to take the detour which is $w_j$ steps longer than the direct path.

	We define a partition of $w_1,\dots,w_b$ into multi-sets $S_1,\dots,S_k$ as follows:
    If path $p_i$ takes the detour between $v_{j-1}$ to $v_j$, we add $w_j$ to $S_i$.
    Now the length of $p_i$ equals $2(b+1)+|S_i|$ for $1 \le i \le k$.
	We can now see that it is possible to choose $p_1,\dots,p_k$ of length $2(b+1) + N/k$ each
    if and only if it is possible to
	partition $w_1,\dots,w_b$ into multi-sets $S_1,\dots,S_k$ of size $N/k$ each.
\end{proof}
}

\short{}{
\begin{figure}[tb]
	\begin{center}
		\begin{tikzpicture}[>=latex,scale=1.0]
		% path
		\foreach \x in {0,1,...,8}{%12}{
			\node[draw, circle, fill=lightgray, scale=0.5, label={[label distance=0.09cm]270:$v_{\x}$}] at (\x,0) (\x) {}; 
		}
		\node[] at (9,0) (9) {};
		\foreach \x in {0.5,1.5,2.5,3.5,4.5,5.5,6.5,7.5,8.5}{
			\pgfmathtruncatemacro\pre{\x-0.5};
			\pgfmathtruncatemacro\suc{\x+0.5};
			\node[draw, circle, fill=lightgray, scale=0.3] at (\x,0.2) (v1) {};
			\node[draw, circle, fill=lightgray, scale=0.3] at (\x,0) (v2) {}; 
			\node[draw, circle, fill=lightgray, scale=0.3] at (\x,-0.2) (v3) {}; 
			\draw[] (\pre) -- (v1) -- (\suc);
			\draw[] (\pre) -- (v2) -- (\suc);
			\draw[] (\pre) -- (v3) -- (\suc);
		}
		% start vertices	
		\draw (-1,0) node[label={[label distance=0.09cm]180:start}] () {};
		\node[draw, circle, fill=lightgray, scale=0.5] at (-1,0.5) (a) {};
		\node[draw, circle, fill=lightgray, scale=0.5] at (-1,0.2) (b) {};
		\node[draw, circle, fill=lightgray, scale=0.5] at (-1,-0.5) (c) {}; 
		\node[draw, circle, fill=lightgray, scale=0.5] at (-1,-0.2) (d) {};
		\foreach \parent in {a,b,c,d}{
			\draw (\parent) -- (0) {};
		}
		% end vertices	
		\draw (11,0) node[label={[label distance=0.09cm]0:end}] () {};
		\node[draw, circle, fill=lightgray, scale=0.5, label={[label distance=0.09cm]270:$v_{b}$}] at (10,0) (10) {}; 
		\node[draw, circle, fill=lightgray, scale=0.5] at (11,0.5) (a) {}; 
		\node[draw, circle, fill=lightgray, scale=0.5] at (11,0.18) (b) {}; 
		\node[draw, circle, fill=lightgray, scale=0.5] at (11,-0.5) (c) {}; 
		\node[draw, circle, fill=lightgray, scale=0.5] at (11,-0.18) (d) {}; 
		\foreach \parent in {a,b,c,d}{
			\draw (\parent) -- (10) {};
		}
		% first loop
		\foreach \x/\name in {1/a,1.25/b,1.5/c,1.75/d}{
			\node[draw, circle, fill=lightgray, scale=0.3] at (0.25,\x) (\name) {};
		}
		\foreach \parent/\child in {0/a,a/b,b/c,c/d}{
			\draw (\parent) -- (\child) {};
		}
		\draw (d) edge[bend left=10] (1) {};
		\draw (1,0.75) node[label={[label distance=0.2cm]180:$3$}] () {};
		% second loop
		\foreach \x/\name in {1/a,1.25/b,1.5/c,1.75/d,2/e,2.25/f,2.5/g,2.75/h,3/i,3.25/j,3.5/k}{
			\node[draw, circle, fill=lightgray, scale=0.3] at (1.25,\x) (\name) {};
		}
		\foreach \parent/\child in {1/a,a/b,b/c,c/d,d/e,e/f,f/g,g/h,h/i,i/j,j/k}{	
			\draw (\parent) -- (\child) {};
		}
		\draw (k) edge[bend left=10] (2) {};
		\draw (2.1,0.75) node[label={[label distance=0.2cm]180:$10$}] () {};
		% third loop
		\foreach \x/\name in {1/a,1.25/b,1.5/c,1.75/d,2/e,2.25/f,2.5/g,2.75/h,3/i}{
			\node[draw, circle, fill=lightgray, scale=0.3] at (2.25,\x) (\name) {};
		}
		\foreach \parent/\child in {2/a,a/b,b/c,c/d,d/e,e/f,f/g,g/h,h/i}{	
			\draw (\parent) -- (\child) {};
		}
		\draw (i) edge[bend left=10] (3) {};
		\draw (3,0.75) node[label={[label distance=0.2cm]180:$8$}] () {};
		% fourth loop
		\foreach \x/\name in {1/a,1.25/b,1.5/c}{%
			\node[draw, circle, fill=lightgray, scale=0.3] at (3.25,\x) (\name) {};
		}
		\foreach \parent/\child in {3/a,a/b,b/c}{%
			\draw (\parent) -- (\child) {};
		}
		\draw (c) edge[bend left=10] (4) {};%
		\draw (4,0.75) node[label={[label distance=0.2cm]180:$2$}] () {};
		% fifth loop
		\foreach \x/\name in {1/a,1.25/b,1.5/c,1.75/d,2/e,2.25/f,2.5/g}{
			\node[draw, circle, fill=lightgray, scale=0.3] at (4.25,\x) (\name) {};
		}
		\foreach \parent/\child in {4/a,a/b,b/c,c/d,d/e,e/f,f/g}{	%
			\draw (\parent) -- (\child) {};
		}
		\draw (g) edge[bend left=10] (5) {};%
		\draw (5,0.75) node[label={[label distance=0.2cm]180:$6$}] () {};%
		% sixth loop
		\foreach \x/\name in {1/a,1.25/b,1.5/c,1.75/d,2/e}{%
			\node[draw, circle, fill=lightgray, scale=0.3] at (5.25,\x) (\name) {};
		}
		\foreach \parent/\child in {5/a,a/b,b/c,c/d,d/e}{	%
			\draw (\parent) -- (\child) {};
		}
		\draw (e) edge[bend left=10] (6) {};%
		\draw (6,0.75) node[label={[label distance=0.2cm]180:$4$}] () {};%
		% seventh loop
		\foreach \x/\name in {1/a,1.25/b,1.5/c,1.75/d,2/e,2.25/f,2.5/g,2.75/h}{
			\node[draw, circle, fill=lightgray, scale=0.3] at (6.25,\x) (\name) {};
		}
		\foreach \parent/\child in {6/a,a/b,b/c,c/d,d/e,e/f,f/g,g/h}{	
			\draw (\parent) edge (\child) {};
		}
		\draw (h) edge[bend left=10] (7) {};
		\draw (7,0.75) node[label={[label distance=0.2cm]180:$7$}] () {};
		% eigth loop
		\foreach \x/\name in {1/a,1.25/b,1.5/c,1.75/d,2/e,2.25/f}{
			\node[draw, circle, fill=lightgray, scale=0.3] at (7.25,\x) (\name) {};
		}
		\foreach \parent/\child in {7/a,a/b,b/c,c/d,d/e,e/f}{	
			\draw (\parent) edge (\child) {};
		}
		\draw (f) edge[bend left=10] (8) {};
		\draw (8,0.75) node[label={[label distance=0.2cm]180:$5$}] () {};
		
		%\draw (8) edge (9,0) {};
		%\draw (8) edge[bend right=30] (9,0)  {};
		%\draw (8) edge[bend right=60] (9,0)  {};
		\draw (10,0) node[label={[label distance=0.16cm]180:$\ldots$}] () {};
		\end{tikzpicture}
	\end{center}
	\caption{The graph constructed from an instance of \mwnp with treewidth two.}\label{seriesParallelReduction}
\end{figure}
}

\section{\PP Parametrized by Path Dependent Attributes}\label{sec:pathDependent}

In the previous section we solved \PP on forests. Since \PP is \np-hard even
for graphs with treewidth 2, we try to find some path dependent parameters to
cope with its difficulty. 
At first, we will restrict the number of paths, then we
will bound the length of each path and finally we consider
the sum of the lengths of all paths.

\paragraph*{Number of Paths}

We denote the number of paths of an instance by $k$.
%Note that \PP instances where the number of paths with non-zero length
%exceed $|E|$ are trivial no-instances. 
%Thus we will only consider meaningful~$k$. 
We start with $k=1$.
Consider an instance where the length of the single path corresponds to the number of vertices in a complete graph $G$.
\begin{observation}\label{obs:epc-nphard-one-path}
	Since Hamiltonian Path is \np-hard, also {\pp} for $k=1$ is \np-hard
\end{observation}
On the other side, for $k=1$ the special case of $\EPC$ becomes easy.
\begin{observation}
	{\EPC} is solvable in polynomial time for $k=1$ by deciding if the input graph is a path of length $l_1$.
\end{observation} 

Unfortunately, for fixed $k\geq 2$ restricting the number of paths is not enough to gain a polynomial time algorithm.  
This holds for \EPC and therefore also for \PP.

\begin{theorem}\label{the:ecp-npcomplete-two-paths}
    Let $k \ge 2$.
	\EPC\ with $k$ paths is \np-complete on $4$-regular graphs.
\end{theorem}
\short{}{
\begin{proof}
We show that \EPC\ is NP-complete for $k=2$ by a reduction from the Hamiltonian Circuit problem in 3-regular graphs \cite{Garey1976hardnessHamCyc}.
Let $G= (V_G,E_G)$ denote a $3$-regular graph with $|V_G|=n$ of an Hamiltonian Circuit Instance, where one arbitrary but fixed edge $e_0\in E_G$ should be contained in the circuit.
We construct a graph $H$ such that $G$ has a Hamiltonian Circuit if and only if $H$ has an exact path cover with two equal-length paths.
We construct $H$ as follows:
We replace every vertex $v_i\in V_G$ by three vertexes $V_H^i=\{v_{v_i,1},v_{v_i,2},v_{v_i,3}\}$ and add three edges such that they form a  clique $K_3$ (see filled vertexes/thick edges in Figures \ref{zoomHamiltonianCircuit}/\ref{zoomHamiltonianCircuit2}). For every edge $e_l=(v_i,v_j) \in E_G\setminus e_0$ we introduce an additional vertex $v_{e_l}$. We add four incident edges to $v_{e_l}$ such that two edges are incident to two different vertexes of $V_H^i$ as well as $V_H^j$ and no vertex $v\in V_H$ has degree greater than four (see unfilled vertexes). For edge $e_0=(v_i,v_j)\in G$ we introduce four vertexes $V_H^0=\{v_{e_0,0},v_{e_0,1},v_{e_0,2},v_{e_0,3}\}$ instead of one vertex $v_{e_0}$ and connect each vertex with one edge to the corresponding remaining vertexes of $V_H^i$ and $V_H^j$ (see Figure \ref{edgeZero} and \ref{zoomEdgeZero}), such that each edge $e\in V_H^0$ is a terminal node.

\begin{figure}[t]
	\begin{center}
		\begin{minipage}{.49\textwidth}
			\centering
			\begin{tikzpicture}[>=latex,scale=1]	
			\useasboundingbox (0,-0.65) rectangle (6,4.5);
			\draw (3,3.5) node[label=above:{$G$}] (g) {};
			%visual triangle
			
			\node[draw, circle, fill=lightgray, scale=0.5, label=above left:{$v_1$}] at (1.04,3.3) (a) {};
			%\draw (1.04,3.3) node[label=above left:{$v_1$}] (a) {};
			%\filldraw (a) circle (0.85mm); 
			\node[draw, circle, fill=lightgray, scale=0.5, label=above right:{$v_2$}] at (4.96,3.3) (b) {};

			\node[draw, circle, fill=lightgray, scale=0.5, label=below:{$v_3$}] at (3,0.1) (c) {};

			\node[draw, circle, fill=red, scale=0.5, label=above:{\textcolor{red}{$v_0$}}] at (3, 2.3) (d) {};
			
			\draw (a) edge[color=blue] node[above] {$e_1$} (d) {};
			\draw (b) edge node[above] {$e_2$} (d) {};
			\draw (c) edge node[right] {$e_3$} (d) {};
			
			%extension
			\draw (0.1, 2.8) node (a1) {};
			\draw (1.04,4.3) node (a2) {}; 
			\draw (a) edge[thin,dotted] (a1) {};
			\draw (a) edge[thin,dotted] (a2) {};
			%extension
			\draw (5.9, 2.8) node (b1) {};
			\draw (4.96,4.3) node (b2) {}; 
			\draw (b) edge[thin,dotted] (b1) {};
			\draw (b) edge[thin,dotted] (b2) {};
			%extension
			\draw (2.2,-0.4) node (c1) {};
			\draw (3.8,-0.4) node (c2) {}; 
			\draw (c) edge[thin,dotted] (c1) {};
			\draw (c) edge[thin,dotted] (c2) {};
			\end{tikzpicture}
			\caption{Component of Graph $G$.}
			\label{zoomHamiltonianCircuit}
		\end{minipage}%
		\begin{minipage}{0.49\textwidth}
			\centering
			\begin{tikzpicture}[>=latex,scale=1]
			
			\useasboundingbox (0,-0.65) rectangle (6,4.5);
			\draw (3,3.5) node[label=above:{$H$}] (h) {};
			%visual triangle
			
			\node[draw, circle, fill=lightgray, scale=0.5, label=below:{$v_{v_2,1}$}] at (4.96, 2.8) (b1) {};
			\node[draw, circle, fill=lightgray, scale=0.5, label=above:{$v_{v_2,2}$}] at (4.5, 3.6) (b2) {};
			\node[draw, circle, fill=lightgray, scale=0.5, label=right:{$v_{v_2,3}$}] at (5.42, 3.6) (b3) {};
			\draw (5.92, 2.8) node[] (b4) {};
			\draw (4.96, 4.4) node[] (b5) {}; 
			
			\node[draw, circle, fill=lightgray, scale=0.5, label=below:{$v_{v_1,3}$}] at (1.04, 2.8) (a1) {};
			\node[draw, circle, fill=lightgray, scale=0.5, label=above:{$v_{v_1,1}$}] at (1.5, 3.6) (a2) {};
			\node[draw, circle, fill=lightgray, scale=0.5, label=left:{$v_{v_1,2}$}] at (0.58, 3.6) (a3) {};			
			\draw (0.06, 2.8) node[] (a4) {};
			\draw (1.04, 4.4) node[] (a5) {}; 			
			
			\node[draw, circle, fill=lightgray, scale=0.5, label=below:{$v_{v_3,1}$}] at (3, -0.4) (d1) {};
			\node[draw, circle, fill=lightgray, scale=0.5, label=right:{$v_{v_3,2}$}] at (3.46, 0.40) (d2) {};
			\node[draw, circle, fill=lightgray, scale=0.5, label=left:{$v_{v_3,3}$}] at (2.54, 0.40) (d3) {};	
			\draw (3.92, -0.4) node[] (d4) {};
			\draw (2.08, -0.4) node[] (d5) {}; 
			
			\node[draw, circle, fill=red, scale=0.5, label=above:{$v_{v_0,1}$}] at (3, 2.8) (c1) {};
			\node[draw, circle, fill=red, scale=0.5, label=right:{$v_{v_0,2}$}] at (3.46, 2.0) (c2) {};
			\node[draw, circle, fill=red, scale=0.5, label=left:{$v_{v_0,3}$}] at (2.54, 2.0) (c3) {};			
			
			\node[draw, circle, color=blue, scale=0.5, label=above right:{\textcolor{blue}{$v_{e_1}$}}] at (2.02, 2.8) (e0) {};
			\node[draw, circle, color=black, scale=0.5, label=above left:{$v_{e_2}$}] at  (4.02, 2.8) (e1) {};
			\node[draw, circle, color=black, scale=0.5, label=right: {$v_{e_3}$}] at  (3, 1.2) (e2) {};
			
			\draw (a1) edge[ultra thick] node[above] {$$} (a2) {};
			\draw (a1) edge[ultra thick] node[above] {$$} (a3) {};
			\draw (a2) edge[ultra thick] node[right] {$$} (a3) {};
			
			\draw (b1) edge[ultra thick] node[above] {$$} (b2) {};
			\draw (b1) edge[ultra thick] node[above] {$$} (b3) {};
			\draw (b2) edge[ultra thick] node[right] {$$} (b3) {};
			
			\draw (c1) edge[ultra thick,color=red] node[above] {$$} (c2) {};
			\draw (c1) edge[ultra thick,color=red] node[above] {$$} (c3) {};
			\draw (c2) edge[ultra thick,color=red] node[right] {$$} (c3) {};
			
			\draw (d1) edge[ultra thick] node[above] {$$} (d2) {};
			\draw (d1) edge[ultra thick] node[above] {$$} (d3) {};
			\draw (d2) edge[ultra thick] node[right] {$$} (d3) {};
			
			\draw (c1) edge[color = blue] node[above] {$$} (e0) {};
			\draw (c3) edge[color = blue] node[above] {$$} (e0) {};
			\draw (a1) edge[color = blue] node[above] {$$} (e0) {};
			\draw (a2) edge[color = blue] node[right] {$$} (e0) {};
			
			\draw (c1) edge node[above] {$$} (e1) {};
			\draw (c2) edge node[above] {$$} (e1) {};
			\draw (b1) edge node[above] {$$} (e1) {};
			\draw (b2) edge node[right] {$$} (e1) {};
			
			\draw (c2) edge node[above] {$$} (e2) {};
			\draw (c3) edge node[above] {$$} (e2) {};
			\draw (d2) edge node[above] {$$} (e2) {};
			\draw (d3) edge node[right] {$$} (e2) {};    	
			
			%extension 
			\draw (a1) edge[thin,dotted] (a4) {};
			\draw (a3) edge[thin,dotted] (a4) {};
			\draw (a2) edge[thin,dotted] (a5) {};
			\draw (a3) edge[thin,dotted] (a5) {};
			%extensio
			\draw (b1) edge[thin,dotted] (b4) {};
			\draw (b3) edge[thin,dotted] (b4) {};
			\draw (b2) edge[thin,dotted] (b5) {};
			\draw (b3) edge[thin,dotted] (b5) {};
			%extension 
			\draw (d1) edge[thin,dotted] (d4) {};
			\draw (d2) edge[thin,dotted] (d4) {};
			\draw (d1) edge[thin,dotted] (d5) {};
			\draw (d3) edge[thin,dotted] (d5) {};
			\end{tikzpicture}
			\caption{Modified component in $H$.}
			\label{zoomHamiltonianCircuit2}
		\end{minipage}
	\end{center}
\end{figure}

\begin{figure}[t]
	\begin{center}
		\begin{minipage}{.49\textwidth}
			\centering
			\begin{tikzpicture}[>=latex,scale=1]	
			\useasboundingbox (0,-0.3) rectangle (4,1);
			%visual triangle
			\node[draw, circle, color=black, scale=0.5, label=left:{$v_i$}] at  (1,0.575) (a) {};
			\node[draw, circle, color=black, scale=0.5,label=right:{$v_j$}] at  (3,0.575) (b) {};
			
			\draw (a) edge node[above] {$e_0$} (b) {};
			
			%extension
			\draw (0,0.0) node (a1) {};
			\draw (0,1.15) node (a2) {}; 
			\draw (a) edge[thin,dotted] (a1) {};
			\draw (a) edge[thin,dotted] (a2) {};
			%extension
			\draw (4,0) node (b1) {};
			\draw (4,1.15) node (b2) {}; 
			\draw (b) edge[thin,dotted] (b1) {};
			\draw (b) edge[thin,dotted] (b2) {};
			\end{tikzpicture}
			\caption{Edge $e_0$ in $G$.}
			\label{edgeZero}
		\end{minipage}%
		\begin{minipage}{0.49\textwidth}
			\centering
			\begin{tikzpicture}[>=latex,scale=1]	
			\useasboundingbox (0,-0.3) rectangle (5,1);
			%visual triangle
			
			\node[draw, circle, fill=lightgray, scale=0.5,label=below:{$v_{{v_i},1}$}] at  (1,0) (a1) {};
			\node[draw, circle, fill=lightgray, scale=0.5,label=above:{$v_{{v_i},2}$}] at  (1,1.15) (a2) {};
			\node[draw, circle, fill=lightgray, scale=0.5,label=left:{$v_{{v_i},3}$}] at  (0,0.575) (a3) {};
			
			\node[draw, circle, fill=lightgray, scale=0.5,label=below:{$v_{{v_j},1}$}] at  (4,0) (b1) {};
			\node[draw, circle, fill=lightgray, scale=0.5,label=above:{$v_{{v_j},2}$}] at  (4,1.15) (b2) {};
			\node[draw, circle, fill=lightgray, scale=0.5,label=right:{$v_{{v_j},3}$}] at  (5,0.575) (b3) {};			
			
			\node[draw, circle, fill=white, scale=0.5,label=above:{$v_{{e_0},0}$}] at  (2,1.15) (e0) {};	
			\node[draw, circle, fill=white, scale=0.5,label=below:{$v_{{e_0},1}$}] at  (2,0) (e1) {};	
			\node[draw, circle, fill=white, scale=0.5,label=below:{$v_{{e_0},2}$}] at  (3,0)  (e2) {};	
			\node[draw, circle, fill=white, scale=0.5,label=above:{$v_{{e_0},3}$}] at  (3,1.15) (e3) {};	
			
			\draw (a1) edge node[above] {} (e1) {};
			\draw (a2) edge node[above] {} (e0) {};
			\draw (b1) edge node[above] {} (e2) {};
			\draw (b2) edge node[above] {} (e3) {};
			
			\draw (a1) edge[ultra thick] node[above] {} (a2) {};
			\draw (a1) edge[ultra thick] node[above] {} (a3) {};
			\draw (a2) edge[ultra thick] node[above] {} (a3) {};
			
			\draw (b1) edge[ultra thick] node[above] {} (b2) {};
			\draw (b1) edge[ultra thick] node[above] {} (b3) {};
			\draw (b2) edge[ultra thick] node[above] {} (b3) {};
			
			%extension
			\draw (0,1.65) node (a4) {};
			\draw (0,-0.5) node (a5) {}; 
			\draw (a2) edge[thin,dotted] (a4) {};
			\draw (a1) edge[thin,dotted] (a5) {};
			\draw (a3) edge[thin,dotted] (a4) {};
			\draw (a3) edge[thin,dotted] (a5) {};
			%extension
			\draw (5,1.65) node (b4) {};
			\draw (5,-0.5) node (b5) {}; 
			\draw (b2) edge[thin,dotted] (b4) {};
			\draw (b1) edge[thin,dotted] (b5) {};
			\draw (b3) edge[thin,dotted] (b4) {};
			\draw (b3) edge[thin,dotted] (b5) {};
			\end{tikzpicture}
			\caption{Four terminal nodes $V_H^0$.}
			\label{zoomEdgeZero}
		\end{minipage}
	\end{center}
\end{figure}
Assume that $G$ has a Hamiltonian Circuit. Since $G$ is 3-regular, one incident edge to every vertex is not covered by this circuit. Observe that these edges correspond to a perfect matching in $G$. In the remaining proof we will denote the edges that are covered by the matching with $E_G^M$ and the ones covered by the circuit $E_G^C$. We rename the vertexes for all $V_H^i$ such that the vertexes which are incident to the matching are those with the smallest index.
We construct two paths $A$ and $B$ in $H$. 
Consider $e_l=(v_i,v_j)\in E_G^M$. Without loss of generality we assign the subpath $(v_{v_i,1}, v_{e_l}, v_{v_i,2}, v_{v_i,3})$ to $A$, and $(v_{v_j,1}, v_{e_l}, v_{v_j,2}, v_{v_j,3})$ to $B$. To cover the cliques $V_H^i$ and $V_H^j$ completely we assign $(v_{v_j,1}, v_{v_j,2}, v_{v_j,3})$ to $A$ and $(v_{v_i,1}, v_{v_i,2}, v_{v_i,3})$ to $B$ (see Figure \ref{subpaths}). 
For $e_l=(v_i,v_j)\in E_G^C$ we connect the ends of all the subpaths in the cliques via $v_{e_l}$ such that all subpaths of $A$ and $B$ are connected. Due to our construction of $H$ this is always possible. Note that $A$ as well as $B$ correspond to a Hamiltonian Circuit in case you ignore the special structure for edge $e_0$.

For edge $e_0$ we have to adapt our strategy. Since $v_{e_0}$ is replaced by four vertexes with degree one, $A$ and $B$ have to start and finish in one of these vertexes. If $e_0\in E_G^M$, the outgoing edges of both endpoints of both paths are incident to the same clique, if $e_0\in E_G^C$, one edge is incident to each clique respectively (see Figure \ref{matchingEdge}/\ref{circuitEdge}). The resulting two paths form an exact path cover of $H$.

\begin{figure}[t]
	\begin{center}
	\begin{minipage}[t]{0.32\textwidth}
			\centering
			\begin{tikzpicture}
			\useasboundingbox (0,-0.3) rectangle (4,1.5);
			%visual triangle
			\node[draw, circle, fill=lightgray, scale=0.5,label=below:{$v_{{v_i},1}$}] at  (1,0) (a1) {};
			\node[draw, circle, fill=lightgray, scale=0.5,label=above:{$v_{{v_i},2}$}] at  (1,1.15) (a2) {};
			\node[draw, circle, fill=lightgray, scale=0.5,label=right:{$v_{{v_i},3}$}] at  (0,0.575) (a3) {};
			\node[draw, circle, fill=lightgray, scale=0.5,label=below:{$v_{{v_j},1}$}] at  (3,0) (b1) {};
			\node[draw, circle, fill=lightgray, scale=0.5,label=above:{$v_{{v_j},2}$}] at  (3,1.15) (b2) {};
			\node[draw, circle, fill=lightgray, scale=0.5,label=left:{$v_{{v_j},3}$}] at  (4,0.575) (b3) {};
			
			\node[draw, circle, fill=white, scale=0.5,label=above:{$v_{e_l}$}] at  (2,0.575) (e0) {};
			
			\draw (a1) edge[ultra thick,color=red] node[above] {} (e0) {};
			\draw (a2) edge[ultra thick,color=red] node[above] {} (e0) {};
			\draw (b1) edge node[above] {} (e0) {};
			\draw (b2) edge node[above] {} (e0) {};
			
			\draw (a1) edge node[above] {} (a2) {};
			\draw (a1) edge node[above] {} (a3) {};
			\draw (a2) edge[ultra thick,color=red] node[above] {} (a3) {};
			
			\draw (b1) edge[ultra thick,color=red] node[above] {} (b2) {};
			\draw (b1) edge[ultra thick,color=red] node[above] {} (b3) {};
			\draw (b2) edge node[above] {} (b3) {};
			
			%extension
			\draw (0,1.65) node (a4) {};
			\draw (0,-0.5) node (a5) {}; 
			\draw (a2) edge[thin,dotted] (a4) {};
			\draw (a1) edge[thin,dotted,color=red] (a5) {};
			\draw (a3) edge[thin,dotted,color=red] (a4) {};
			\draw (a3) edge[thin,dotted] (a5) {};
			%extension
			\draw (4,1.65) node (b4) {};
			\draw (4,-0.5) node (b5) {}; 
			\draw (b2) edge[thin,dotted,color=red] (b4) {};
			\draw (b1) edge[thin,dotted] (b5) {};
			\draw (b3) edge[thin,dotted] (b4) {};
			\draw (b3) edge[thin,dotted,color=red] (b5) {};
			\end{tikzpicture}
			\caption{Paths \textcolor{red}{$\mathbf{A}$}/$B$.}
			\label{subpaths}
		\end{minipage}
		\begin{minipage}[t]{0.32\textwidth}
			\centering
			\begin{tikzpicture}	
			%visual triangle
			\useasboundingbox (0,-0.3) rectangle (4,1.5);
			\node[draw, circle, fill=lightgray, scale=0.5,label=below:{$v_{{v_i},1}$}] at  (1,0) (a1) {};
			\node[draw, circle, fill=lightgray, scale=0.5,label=above:{$v_{{v_i},2}$}] at  (1,1.15) (a2) {};
			\node[draw, circle, fill=lightgray, scale=0.5,label=right:{$v_{{v_i},3}$}] at  (0,0.575) (a3) {};
			\node[draw, circle, fill=lightgray, scale=0.5,label=below:{$v_{{v_j},1}$}] at  (3,0) (b1) {};
			\node[draw, circle, fill=lightgray, scale=0.5,label=above:{$v_{{v_j},2}$}] at  (3,1.15) (b2) {};
			\node[draw, circle, fill=lightgray, scale=0.5,label=left:{$v_{{v_j},3}$}] at  (4,0.575) (b3) {};
			
			\node[draw, circle, fill=white, scale=0.5] at  (1.6,1.15) (e0) {};
			\node[draw, circle, fill=white, scale=0.5] at  (1.6,0) (e1) {};
			\node[draw, circle, fill=white, scale=0.5] at  (2.4,0) (e2) {};
			\node[draw, circle, fill=white, scale=0.5] at  (2.4,1.15)  (e3) {};			
			
			\draw (a1) edge[ultra thick,color=red] node[above] {} (e1) {};
			\draw (a2) edge[ultra thick,color=red] node[above] {} (e0) {};
			\draw (b1) edge node[above] {} (e2) {};
			\draw (b2) edge node[above] {} (e3) {};
			
			\draw (a1) edge node[above] {} (a2) {};
			\draw (a1) edge[ultra thick,color=red] node[above] {} (a3) {};
			\draw (a2) edge node[above] {} (a3) {};
			
			\draw (b1) edge[ultra thick,color=red] node[above] {} (b2) {};
			\draw (b1) edge node[above] {} (b3) {};
			\draw (b2) edge[ultra thick,color=red] node[above] {} (b3) {};
			
			%extension
			\draw (0,1.65) node (a4) {};
			\draw (0,-0.5) node (a5) {}; 
			\draw (a2) edge[thin,dotted,color=red] (a4) {};
			\draw (a1) edge[thin,dotted] (a5) {};
			\draw (a3) edge[thin,dotted] (a4) {};
			\draw (a3) edge[thin,dotted,color=red] (a5) {};
			%extension
			\draw (4,1.65) node (b4) {};
			\draw (4,-0.5) node (b5) {}; 
			\draw (b2) edge[thin,dotted] (b4) {};
			\draw (b1) edge[thin,dotted,color=red] (b5) {};
			\draw (b3) edge[thin,dotted,color=red] (b4) {};
			\draw (b3) edge[thin,dotted] (b5) {};
			\end{tikzpicture}
			\caption{Edge $e_0\in E_G^M$.}
			\label{matchingEdge}
		\end{minipage}
		\begin{minipage}[t]{0.32\textwidth}
			\centering
			\begin{tikzpicture}	
			%visual triangle
			\useasboundingbox (0,-0.3) rectangle (4,1.5);
			\node[draw, circle, fill=lightgray, scale=0.5,label=below:{$v_{{v_i},1}$}] at  (1,0) (a1) {};
			\node[draw, circle, fill=lightgray, scale=0.5,label=above:{$v_{{v_i},2}$}] at  (1,1.15) (a2) {};
			\node[draw, circle, fill=lightgray, scale=0.5,label=right:{$v_{{v_i},3}$}] at  (0,0.575) (a3) {};
			\node[draw, circle, fill=lightgray, scale=0.5,label=below:{$v_{{v_j},1}$}] at  (3,0) (b1) {};
			\node[draw, circle, fill=lightgray, scale=0.5,label=above:{$v_{{v_j},2}$}] at  (3,1.15) (b2) {};
			\node[draw, circle, fill=lightgray, scale=0.5,label=left:{$v_{{v_j},3}$}] at  (4,0.575) (b3) {};
			
			\node[draw, circle, fill=white, scale=0.5] at  (1.6,1.15) (e0) {};
			\node[draw, circle, fill=white, scale=0.5] at  (1.6,0) (e1) {};
			\node[draw, circle, fill=white, scale=0.5] at  (2.4,0) (e2) {};
			\node[draw, circle, fill=white, scale=0.5] at  (2.4,1.15)  (e3) {};	
			
			\draw (a1) edge[ultra thick,color=red] node[above] {} (e1) {};
			\draw (a2) edge node[above] {} (e0) {};
			\draw (b1) edge node[above] {} (e2) {};
			\draw (b2) edge[ultra thick,color=red] node[above] {} (e3) {};
			
			\draw (a1) edge[ultra thick,color=red] node[above] {} (a2) {};
			\draw (a1) edge[ultra thick] node[above] {} (a3) {};
			\draw (a2) edge[ultra thick,color=red] node[above] {} (a3) {};
			
			\draw (b1) edge[ultra thick,color=red] node[above] {} (b2) {};
			\draw (b1) edge[ultra thick,color=red] node[above] {} (b3) {};
			\draw (b2) edge node[above] {} (b3) {};
			
			%extension
			\draw (0,1.65) node (a4) {};
			\draw (0,-0.5) node (a5) {}; 
			\draw (a2) edge[thin,dotted] (a4) {};
			\draw (a1) edge[thin,dotted] (a5) {};
			\draw (a3) edge[thin,dotted] (a4) {};
			\draw (a3) edge[thin,dotted,color=red] (a5) {};
			%extension
			\draw (4,1.65) node (b4) {};
			\draw (4,-0.5) node (b5) {}; 
			\draw (b2) edge[thin,dotted] (b4) {};
			\draw (b1) edge[thin,dotted] (b5) {};
			\draw (b3) edge[thin,dotted,color=red] (b4) {};
			\draw (b3) edge[thin,dotted] (b5) {};
			\end{tikzpicture}
			\caption{Edge $e_0\in E_G^C$}
			\label{circuitEdge}
		\end{minipage}
	\end{center}
\end{figure}

Assume $H$ has an exact path cover with two paths. If the endpoints of the paths are incident to the same clique, the corresponding edge in $G$ is not contained in the Hamiltonian Circuit and vice versa. All other vertexes have degree four, so they have to be traversed by both paths. Let $H_i'$ denote the subgraph that contains the clique that corresponds to vertex $v_i\in G$ and the three vertexes that represent the incident edges. $H_i'$ has to be entered and left once by each path. So for every subgraph one incident edge-vertex is only traversed by one path. Since $H$ has an exact path cover, there exists a second subgraph where this vertex is also contained in the other path. Thus the corresponding edge in $G$ is not contained in the Hamiltonian Circuit. Since $G$ is $3$-regular, all remaining edges  form a Hamiltonian Circuit.\\
For $k>2$ one can add $k-2$ paths of length one to $P$ as well as $k-2$ isolated edges to $G$. Since this edges can only be covered by the new paths the proof above also holds for all $k>2$.
\end{proof}
}

\paragraph*{Paths with bounded length}

Observe that all hardness proofs that we have seen so far somehow involve paths of a
certain length. Thus, we analyze the complexity of \PP based on the length of
the paths we want to pack.

\medskip
\begin{tabularx}{0.95\textwidth}{ r X p{0.5cm} }
	\multicolumn{2}{l}{\PIP} \\
	Input: & A set of paths \( P = \{p_{1},\dots,p_{k} \} \) of length \(l_1=\ldots=l_k=i\), a graph \(G=(V,E)\). \\
	Question: & Is there an edge-disjoint embedding of $P$ into $G$ \suchthat each edge $e\in E$ is covered exactly once?
\end{tabularx}
\medskip

\ptwop is solvable in polynomial time by reformulating it as a matching problem on the line graph.
%Wproblem is \np-hard even for length $i = 3$.  
We show that \pthreep is already \np-hard
via a reduction from the 3-dimensional matching problem, one of Karp's original 21 NP-complete problems.
The 3-dimensional matching problem takes as input
sets $X,Y,Z$ of size $n$ and $T \subseteq X \times Y \times Z$.
The question is whether there exists a set $M \subseteq T$ of size $n$
such that for all $(x_1,y_1,z_1), (x_2,y_2,z_2) \in M$,
$x_1 \neq x_2$, $y_1 \neq y_2$, $z_1 \neq z_2$.
The reduction is similar to the $P_2$-packing reduction in \cite{pantel2001graph}.

\begin{theorem}\label{the:pip-nphard}
    \pthreep is \textnormal{NP}-hard.
\end{theorem}
\short{}{
\begin{proof}
    We reduce from 3-dimensional matching.
    Let $X,Y,Z,T$ be a 3-dimensional matching instance.
    For each $t \in X \cup Y \cup Z$ we add
    two vertices $t,t'$ and an edge $(t,t')$.
    Furthermore, for every $(t_1,t_2,t_3) \in T$, we attach a gadget to
    the vertices $t_1,t_2,t_3$ as depicted in Figure~\ref{fig:length-3a} and \ref{fig:length-3b}.
    We show that the resulting graph $G$ admits a packing of $P$
    such that all edges are covered
    if and only if $X,Y,Z,T$ admits a 3-dimensional matching.

    Consider a packing of $P$ in $G$ such that all edges are covered.
    Every $p_j\in P$ contains edges from at most one gadget.
    There are only two ways in which a gadget can be covered.
    The first way (left) covers $(t_1,t_1')$, $(t_2,t_2')$, $(t_3,t'_3)$ and the second way (right)
    does not cover $(t_1,t_1')$, $(t_2,t_2')$, $(t_3,t'_3)$.
    We define $M \subseteq T$
    to be the set of all tuples $(t_1,t_2,t_3)$ such that
    there exists a gadget whose paths cover
    $(t_1,t_1')$, $(t_2,t_2')$, $(t_3,t'_3)$.
    Every edge in $G$ is covered exactly once, therefore $M$
    is a 3-dimensional matching.

    On the other hand, assume $M$ to be a 3-dimensional matching.
    We construct a packing as follows:
    Let $(t_1,t_2,t_3) \in T$.
    We consider the gadget which is attached to $t_1,t_2,t_3$.
    If $(t_1,t_2,t_3) \in M$, we pack the paths as depicted on the left.
    Otherwise, we pack the paths as depicted on the right.
    In the resulting packing every edge is covered exactly once.
\end{proof}
}
\short{}{
\begin{figure}[tb]
	\tikzstyle{vertex}=[circle,fill=black,scale=0.3]
	\tikzstyle{edge} = [draw,line width=1.1pt,-]
	\tikzstyle{edgedot} = [draw,line width=1.1pt,-,dotted]
	\begin{center}
		\begin{minipage}{.49\textwidth}
			\centering
			\begin{tikzpicture}[>=latex,scale=0.6]
			% Create vertices
			\foreach \pos/\name in {
				{(0,0)/a0}, {(0,1)/a1}, {(0,2)/a2}, {(0,3)/a3}, {(0,4)/a4}, {(0,5)/a5},
				{(2,0)/b0}, {(2,1)/b1}, {(2,2)/b2}, {(2,3)/b3}, {(2,4)/b4}, {(2,5)/b5},
				{(3.5,0)/gap},
				{(5,0)/c0}, {(5,1)/c1}, {(5,2)/c2}, {(5,3)/c3}, {(5,4)/c4}, {(5,5)/c5},
				{(-1,2)/adongle},
				{(1,2)/bdongle},
				{(4,2)/cdongle}}
			\node[draw, circle, fill=lightgray, scale=0.5] (\name) at \pos {};
			\node () at (-0.7,4) {$t_1$};
			\node () at (1.3,4) {$t_2$};
			\node () at (4.3,4) {$t_3$};
			\node () at (-0.7,5) {$t'_1$};
			\node () at (1.3,5) {$t'_2$};
			\node () at (4.3,5) {$t'_3$};
			% Connect vertices with edges
			\foreach \source/ \dest in {
				a2/a3, a3/a4, a4/a5,
				b2/b3, b3/b4, b4/b5,
				c2/c3, c3/c4, c4/c5,
				a0/b0,b0/gap, gap/c0}
			\path[edge, color = red] (\source) -- node {} (\dest);
			\foreach \source/ \dest in {
				a0/a1, a1/a2, a2/adongle,
				b0/b1, b1/b2, b2/bdongle,
				c0/c1, c1/c2, c2/cdongle}
			\path[edge, color = blue, ] (\source) -- node {} (\dest);
			
			\end{tikzpicture}
			\caption{Optimal path packing if \\$(t_1,t_2,t_3)$ is part of the matching.}
            \label{fig:length-3a}
		\end{minipage}%
		\begin{minipage}{0.49\textwidth}
			\centering
			\begin{tikzpicture}[>=latex,scale=0.6]
			% Create vertices
			\foreach \pos/\name in {
				{(0,0)/a0}, {(0,1)/a1}, {(0,2)/a2}, {(0,3)/a3}, {(0,4)/a4}, {(0,5)/a5},
				{(2,0)/b0}, {(2,1)/b1}, {(2,2)/b2}, {(2,3)/b3}, {(2,4)/b4}, {(2,5)/b5},
				{(3.5,0)/gap},
				{(5,0)/c0}, {(5,1)/c1}, {(5,2)/c2}, {(5,3)/c3}, {(5,4)/c4}, {(5,5)/c5},
				{(-1,2)/adongle},
				{(1,2)/bdongle},
				{(4,2)/cdongle}}
			\node[draw, circle, fill=lightgray, scale=0.5] (\name) at \pos {};
			\node () at (-0.7,4) {$t_1$};
			\node () at (1.3,4) {$t_2$};
			\node () at (4.3,4) {$t_3$};
			\node () at (-0.7,5) {$t'_1$};
			\node () at (1.3,5) {$t'_2$};
			\node () at (4.3,5) {$t'_3$};
			% Connect vertices with edges
			\foreach \source/ \dest in {
				a2/a3, a3/a4, a2/adongle,
				b0/b1, b1/b2, b0/gap,
				c2/c3, c3/c4, c2/cdongle}
			\draw[edge, color = red] (\source) -- node {} (\dest);
			\foreach \source/ \dest in {
				b2/b3, b3/b4, b2/bdongle,
				a0/a1, a1/a2, a0/b0,                	 
				c0/c1, c1/c2, c0/gap}
			\draw[edge, color = blue] (\source) -- node {} (\dest);
			\foreach \source/ \dest in {a4/a5, b4/b5, c4/c5}
			\path[edgedot] (\source) -- node {} (\dest);
			\end{tikzpicture}
			\caption{Optimal path packing if \\$(t_1,t_2,t_3)$ is not part of the matching.}
            \label{fig:length-3b}
		\end{minipage}
	\end{center}
\end{figure}
}

\paragraph*{Bounded sum of path lengths}

The two previous results show that \PP is \np-hard even if the number of paths or the maximal
length of the paths is bounded by a constant.
At last, we the problem is in \fpt when parameterized by
the number of paths \emph{and} their length.
We give a randomized \fpt algorithm using color coding~\cite{parameterizedalgorithms,colorcoding}
that can easily be derandomized using perfect hash families~\cite{parameterizedalgorithms}.

\begin{theorem}\label{the:fptsummedpathlength}
	{\PP} parameterized by the summed length of all paths is in \fpt.
\end{theorem}
\short{}{
\begin{proof}
	Let $G$, $p_1,\dots,p_k$ be a {\pp} instance.
	The parameter $l$ is the summed length of $p_1,\dots,p_k$.
	We color the edges of $G$ with $k$ colors.
	Let $G_i$ be $G$ induced on all edges with color~$i$.
	We return yes if for $i \in \{1,\dots,k\}$, $p_i$ can be embedded into $G_i$.
	This can be done in FPT time~\cite{parameterizedalgorithms}.
	If $G,p_1,\dots,p_k$ is a no-instance then this algorithm returns no.
	If $G,p_1,\dots,p_k$ is a yes-instance
	then $p_1,\dots,p_k$ can be embedded edge-disjoint into $G$.
	With probability at least $l^{-l}$ the embedding of $p_i$ is colored with color $i$
	and the algorithm returns yes.
	We repeat the procedure $\Omega(l^l)$ times to obtain the correct
	answer with probability at least $2/3$.
	This yields a randomized FPT algorithm for {\pp}.
\end{proof}
}

For packing vertex-disjoint paths similar results are known:
The $P_2$-packing problem takes a graph $G$ and an integer $k$
and asks if there is a set of $k$ vertex-disjoint $P_2$ in $G$.
This problem is NP-complete~\cite{kirkpatrick1978completeness, pantel2001graph}.
Fernau and Raible given FPT algorithm parameterized
by the number of paths~\cite{fernau2008parameterized}.

\section{\PP Parametrized by Graph Dependent Attributes}\label{sec:bcd}

% relaxation to not exact?

Earlier (\cref{the:epc-nphard-forests}), we showed that \EPC is \np-hard even on a single subdivided star.
So even for trees where there is only one node of degree higher than two the problem becomes \np-hard.
In this section we study further restrictions to forests and finally identify a polynomial time solvable case.
We do so by considering restrictions to the following three parameters:
number of vertices of degree at least three, the maximal degree, and the number of connected components. 
For an easier notation we define this combined parameter as the `bcd' of graph \( G \). 
It is a bound on \textbf{b}ranching nodes, connected \textbf{c}omponents and maximum \textbf{d}egree.

\begin{definition}
	Let \( G \) be a graph.
	Then \( \bcd(G) \) is the minimal \( k \in \N \) such that \(
	G \) has at most \( k \) nodes of degree larger than two, at
	most \( k \) connected components, and a maximal degree of at most \( k \).
\end{definition}

The above mentioned reduction showing \np-hardness for a subdivided star constructs a graph with unbounded degree.
What is the complexity if we limit the vertex degree to a constant, but in turn allow multiple components?
Unfortunately even for a forest of paths, thus a maximum degree of two, the problem remains \np-hard.
\np-hardness follows by an easy adaption of the proof of \cref{the:epc-nphard-forests}.
The constructed subdivided star has an even number \( 2m \) of legs of length \( \ell \).
%Thus instead of packing some paths into two legs of length \( \ell \), we might pack them into a path of length \( 2 \ell \).
Instead one could also use \( m \) disjoint paths of length \( 2\ell \).
Thus we can follow \np-hardness of \EPC even for forests of paths.

\begin{corollary}\label{col:NP-on-forests-of-paths}
	\EPC is \np-hard even on forests of paths.
\end{corollary}
Thus, if we drop either the degree or the number of components as parameters,
the problem becomes \np-complete, even if the remaining parameters are bounded
by a constant.
Thus the remaining question is: What is the complexity if we limit the vertex
degree to a constant, limit the number of connected components to a constant,
but in turn allow arbitrary many vertices of degree at least three?
We show hardness in this scenario even for the more restricted problem
of packing paths of ascending length, denoted by \PPP.

%We show that \EPC is \np-hard for the following even more restricted case.
%Let \PPP be the problem of edge-disjoint packing paths of ascending lengths \( 1,2,\dots \) into a graph $G$.
%There the task is to pack a graph with a set of paths of length ascending from one, thus paths of length \( 1,2,\dots \) up to some \( n \in \N \).
%We denote this problem as \PPP.
%There the task is to pack a graph with a set of paths of length ascending from one, thus paths of length \( 1,2,\dots \) up to some \( n \in \N \).
%We denote this problem as \PPP.
%To avoid trivial instances, we may require that \( |E(G)| \) is equal to \( 1+2+\dots+n \) for an \( n \in \N \).
%There the task is to pack $k$g  paths \Pperfect$=\{p_1,\ldots,p_k\}$ with $l_1=1,l_2=2,\ldots,l_k=k$.

\medskip
	\begin{tabularx}{0.95\textwidth}{ r X p{0.5cm} }
		\multicolumn{2}{l}{\PPP} \\
		Input: & A graph \(G=(V,E)\) where \( |E(G)| = 1+ 2 + \dots + n \) for some \( n \in \N \). \\
		Question: & Does there exist an edge-disjoint embedding of paths \( p_1,\dots,p_n \) with lengths
		\( \ell_1=1,\dots,\ell_n=n \) into $G$ \suchthat each edge $e\in E$ is covered exactly once?
	\end{tabularx}
	
\medskip
We show \np-hardness of this restricted problem on subdivisions of a caterpillar with vertex degree at most eight.
We reduce from the following unary version of 3-partition.
This version is slightly non-standard since we require that no numbers occur twice and 
relax the condition to put exactly three elements into each partition.
%In particular, the input of integers is a set, so no number occurs twice, and moreover there is no explicit condition to put exactly three sets into each partition.

\def\upp{\textsc{Unary 3-Partition}}

\medskip	
\begin{tabularx}{0.95\textwidth}{ r X p{0.5cm} }
	\multicolumn{2}{l}{\upp} \\
	Input: & A set of integers \( A = \{a_1,\dots,a_{3s} \} \subseteq \N  \) in unary encoding. \\
	Question: & Is there a partition of \( A \) into \( s \) sets of equal sum?
\end{tabularx}
\medskip

Hulett et al.~show that the above problem is \np-hard if we require each of the \( s \) partitions to contain exactly three elements~\cite{HULETT2008594}.
We get \np-hardness without the extra condition by increasing each number in \( A \) by adding a big number (for example \( \sum_{i=}^{3s} a_i \)).

We sketch how to reduce from \PPP on caterpillars with vertex degree at most eight to \upp.
Note that, each partition of a \upp\ instance must have size \( m = \frac{1}{s} \sum_{i=1}^{3s} a_i \).
Consider the paths \( \{ p_i \; | \; 1 \leq i \leq m, \; i \in A \} \) whose length occurs in \( A \).
We translate a partition of $A$ into an exact packing of these paths.
However, we have to account for the paths \( \{ p_i \; | \; 1 \leq i \leq m, \; i \notin A \} \) whose length occurs \textit{not} as an integer in set \( A \).
For each of these we introduce a path where it fits in precisely, and by an exchange argument we may assume it is packed there.
Now, roughly what remains to do is to construct a large caterpillar of low maximum degree where all these paths can be packed in.

\begin{theorem}\label{the:ppc-npcomplete-caterpillars}
	\PPP is \np-hard even for subdivided caterpillars with vertex degree at most eight.
\end{theorem}
\short{}{
\begin{proof}
	We show \np-hardness by a polynomial time reduction from \upp.
	
	\newcommand{\pa}{\mathcal{P}_A}
	\newcommand{\pan}{\mathcal{P}_{\bar{A}}}
	\newcommand{\p}{\mathcal{P}}
	\newcommand{\pmm}{\mathcal{P}_{\mathcal{M}}}
	Let the set \( A = \{a_1,\dots,a_{3s}\} \subseteq \N \) be a \upp\ instance.
	Let \( m := \frac{1}{s} \sum_{i=1}^{3s} a_i \) be the `bucket size', in other words the target size of the sets of the partition.
	We may assume that \( s \geq 2 \) and \( a_i \leq m \) for every \( i \in \N \).
	We construct a \PPP instance with a maximal path length of \( n := 8m \).
	Let us denote by \( p_i \) a path of length \( i \) for \( i \in \N \) .
	In our construction, we use paths of lengths equal to the integers that do \textit{not} occur in \( A \), which is \( \pan = \{ p_i \; | \; 1 \leq i \leq m, \; i \notin A \} \).
	Further, we use a set of paths whose length is longer than any integer in \( A \), more specifically \( \p = \{ p_{m+1}, p_{m+2}, \dots, p_n \} \).
	Moreover, let \( \pmm \) be a set of \( s \) many paths of length \( m \).
	Note, that since \( a_1, \dots, a_{3s} \) are distinct, \( s < m \).
	We construct a \PPP instance by combining the paths \( \pmm, \pan \) and \( \p \) into a subdivided caterpillar \( T \) of degree at most eight.
	This then is a polynomial time reduction since the integers in \( A \) are given in unary.
	
	For a yes-instance of \upp, paths \( \{ p_i \; | \; 1 \leq i \leq m, \; i \in A \} \) may be embedded into \( \pmm \) covering every edges exactly once.
	Then the paths \( \p \) and \( \pan \) may be straight-forwardly embedded into the remaining paths in \( T \) of length \( 1,\dots,n \).
	Thus a yes-instance of \upp\ implies a \ppp of \( T \), independent of the precise construction.
	
	Let us specify the construction.
	Let path \( p_n \) consists of the vertex sequence \( v_0, \dots, v_n = v_{8m} \).
	In the following we identify a certain vertex \( v \) of \( p_n \) with certain vertex \( u \) of another path; thereby constructing a larger graph where \( u=v \).
	We denote by the \textit{middle} of a path an arbitrary vertex that maximizes the distance to its leaves (which is either unique or there are two choices).
	Identify vertex \( v_{4m} \) of path \( p_n \) with the middle of path \( p_{n-1} \).
	Identify vertex \( v_{4m} \) of \( p_n \) with the middle of \( p_{n-2} \) and the middle of \( p_{n-3} \), then \( v_{m+1} \) with the middle of \( p_{n-4} \) and the middle of \( p_{n-5} \) and so on;
	In other words, for \( i = n-2, \dots, m+1 \), identify the middle of \( p_i \) with \( v_{j_i} \) where \( j_i := 4m + \lfloor \frac{ (n-2)-i }{2} \rfloor \).
	Then, for \( p_i \in \pan \), identify the middle of \( p_i \) with \( v_{4m - i} \).
	Finally, identify the middle vertices of the paths \( \pmm \) with \( v_{3m-1}, \dots, v_{3m-s} \), respectively.
	Note that \( 3m-s \geq 2m \).
	This way we constructed a caterpillar \( T \) with maximum degree eight at vertex \( v_{4m} \) and using the paths \( \pmm, \pan \) and \( \p \).
		
	It remains to show that if the constructed \PPP instance has a perfect packing of paths \( q_1, \dots, q_n \) of length \( 1,\dots,n \) respectively, then also \( A \) has a partition into \( s \) sets each of sum \( m \).
	We claim that path \( q_n \) may only cover the edges of \( p_n \) and \( p_{n-1} \).
	Any other path \( p_i \) for \( i < n-1 \) has length \( < n \), and either side of path \( p_i \) from \( v_{j_i} \) has length \( \leq \lceil \frac{i}{2} \rceil \).
	The longest path starting in \( v_{j_i} \) is the path \( v_{j_i}, v_{j_i-1}, \dots, v_0 \) of length \( j_i \).
	It is easy to see that \( j_i + \lceil \frac{i}{2} \rceil < n \), thus path \( q_n \) does not fit into \( p_i \). % tbh handwavy
	
	Because of symmetry, we may assume that \( q_n \) and \( q_{n+1} \) exactly cover the edges of \( p_n \) and \( p_{n-1} \).
	Inductively, for \( i = n-2, n-4, \dots, 3s+1 \), paths \( q_i, q_{i-1} \) may only cover the edges of \( p_i, p_{i-1} \).
	What remains to be covered are \( \pan \) and \( \pmm \) by paths \( q_1, \dots, q_m \).
	
	Consider the longest path \( p_j \in \pan \) that is not covered by the path of equal length \( q_j \) but instead by at least two paths, say by the paths \( q_1', q_2',\dots \).
	Then path \( q_j \) has to occur completely in a path \( p \) of \( \pmm \).
	We may exchange path \( q_j \) with paths \( q_1', q_2',\dots \), such that \( p_j \) is covered by \( q_j \), and \( q_1', q_2',\dots \) occur in \( p \).
	Repeat this exchange argument, until \( \pan \) is covered exactly by \( \{ q_i \; | \; 1 \leq i \leq m, \; i \notin A \} \).
	Then the paths \( \mathcal{Q}_A := \{ q_i \; | \; 1 \leq i \leq m, \; i \in A \} \) must exactly cover \( \pmm \).
	The mapping of \( \mathcal{Q}_A \) on \( s \) many paths \( \pmm \), each of length \( m \), defines a partition of \( A \) into \( s \) sets, each of size \( m \).
	Thus the set \( A \) is a yes-instance of \upp.
\end{proof}
}

The maximum degree eight in the theorem above was chosen to simplify the proof.
One can show \np-hardness even for smaller maximum vertex degree.

\medskip

\paragraph*{XP-algorithm for parameter \( \bcd\boldsymbol ( \boldsymbol G \boldsymbol ) \)}

We saw that if two bcd-parameters are constant and one bcd-parameter is unbounded
then \EPC is \np-complete.
We further study the complexity when parameterized by all three parameters.
We give an XP-algorithm for \PP parametrized by \( \bcd(G) \).
This means for every fixed \( k \), there is a polynomial time algorithm for graphs with \( \bcd(G) \leq k \).

\begin{theorem}\label{the:pp-upper-bound-fpt}
    There is a \( k!^k (n+k^2)^{O(k^2)} \)-time algorithm for \PP with \( k = \bcd(G) \).
\end{theorem}
%\short{}{
\begin{proof}
	We give an algorithm, that given a graph \( G \) and a list \( P \) of paths \( p_1,\dots,p_k \), decides if there is an edge-disjoint embedding of \( p_1,\dots,p_k \) into \( G \).
	To do so, we guess a partition of \( G \) into eventually a set of vertex-disjoint paths \( \mathcal{X} \).
	Then it suffices to find an embedding of \( P \) into such a set of vertex-disjoint paths \( \mathcal{X} \).
	The remaining problem then is just a generalized bin-packing problem with \( O(k^2) \) bins, but encoded in unary; thus solvable in time \( n^{O(k^2)} \).
	Most technicality lies in guessing the vertex-disjoint paths \( \mathcal{X} \). 
	First we guess a partition into a bounded number of walks \( \mathcal{W} \).
	Later we need to partition \( \mathcal{W} \) further resulting in vertex-disjoint paths \( \mathcal{X} \).
	
%	First, if there are connected components which are a cycle, we replace this cycle by a path of the same number of edges.
%	This modification does not change whether there is an embedding, and also does not change the parameter \( k \).
%	
	Let \( V_1 \) be the set of vertices of degree two.
	Let \( V_2^\star \) consist of a vertex of every connected component that is a circle.
	Let \( V_2 \) be the set of vertices of degree two that are not in \( V_2^\star \), and let \( V_{\geq 3} \) be the vertices of degree at least three including \( V_2^\star \).
	This seemingly odd definition allows us to work with walks starting and ending in \( V_1 \cup  V_{\geq 3} \) that cover every edge, in particular those in a circle.
	Because there are at most \( k \) connected components, \( |V_2^\star| \leq k \).
	Then since there are at most \( k \) vertices of degree at least three, we have \( |V_{\geq 3}| \leq 2k \).
%	We distinguish between vertices of degree \( 1, 2 \) and \( \geq 3 \), and denote the set of these vertices as \( V_1, V_2 \) and \( V_{\geq 3} \) respectively,
		
	Assuming a yes-instance, there is an edge-disjoint embedding of paths \( P \) into graph \( G \).
	At every vertex \( v \in V_{\geq 3} \) every path of \( P \) contains at most two of the incident edges of \( v \).
	Thus at every vertex \( v \in V_{\geq 3} \) there is a maximal matching \( M_v \) of \( v \)'s incident edges such that no path in \( P \) contains two unmatched edges.
	
	% comment: this is also called a trail, but probably this notion is not so common (?)
	We consider `direct' walks between `neighboring' vertices \( V_1 \cup V_{\geq 3} \):
	Let \( \mathcal{Q} \) be the set of walks between \( u,v \in V_1 \cup V_{\geq 3} \) with inner vertices from \( V_2 \),
		and further where no vertex among \( V_2 \) is repeated (though possibly \( u=v \)).
%	Since no connected component of \( G \) is a cycle, every edge of \( G \) is covered by exactly one such walk \( \mathcal{Q} \).
	We join these walks \( \mathcal{Q} \) to a set of walks \( \mathcal{W} \) according to matchings \( M_v \) for \( v \in V(G) \).
	Whenever two walks \( w_1, w_2 \) end at some edges \( uv \) respectively \( u'v \), and \( uv \) is matched to \( u'v \) by \( M_v \), then join walks \( w_1 \) and \( w_2 \) at edges \( uv,vu' \).
	This procedure terminates and yields a well defined set of walks \( \mathcal{W} \).
	
	Note that every edge is covered by a walk \( \mathcal{Q} \) and thus also every edge is covered by a walk \( \mathcal{W} \).
	We further claim that every path \( p \) of \( P \) is a subsequence of edges of some walk \( w \in \mathcal{W} \).
		Assuming otherwise, there are walks \( w_1, w_2 \) ending at edges \( uv \) and \( u'v \).
		Then \( v \) is not a leaf, and thus \( v \in V_{\geq 3} \).
	Then matching \( M_v \) matches edges \( uv \) and \( u'v \), and thus \( w_1, w_2 \) had to be joined to a single walk.
	
	Thus for a yes-instance there is at least one set of matchings \( M_v, v \in V_{\geq 3} \) which determines walks \( \mathcal{W} \) such that \( P \) may be embedded into \( \mathcal{W} \).	
	An algorithm may try the possible partition of edges into such a set of walks \( \mathcal{W} \) as follows.
	Guess for each vertex \( v \in V_{\geq 3} \) a maximal matching \( M_v \) of its incident edges.
	There are at most \( k \) high degree vertices \( V_{\geq 3} \setminus V_2^\star \), each with at most \( k \) incident edges.
		(Also we have a matching for \( V_2^\star \), though since there are only two incident edges, there is only one possible matching.)
	Thus the algorithm tries at most \( k!^k \) possibilities.
	Then combine the paths \( \mathcal{Q} \) to walks \( \mathcal{W} \) according to the matchings, which is possible in polynomial time.
	
	We claim that \( \mathcal{W} \) has at most \( k^2 \) walks.
	Every walk in \( \mathcal{W} \) has two endpoints, and the endpoints are among \( V_1 \cup V_{\geq 3} \).
	Clearly, at every leaf \( v \in V_1 \) at most one path ends.
		Further, there are at most \( k^2 \) leaves in the input graph of \( \leq k \) vertices of degree \( \geq 3 \) and maximal degree of \( k \).
	If at a vertex \( v \in V_{\geq 3} \) two walks \( w_1, w_2 \) end, there are edges \( uv \) of \( w_1 \) and \( u'v \) of \( w_2 \) unmatched by \( M_v \), in contradiction to a maximal matching \( M_v \).
%	, and these edges are potentially matched by the maximal matching \( M_v \).
%		Then at least one of the edges \( uv \) or \( u'v \) is matched, which contradicts that walks \( w_1 \) and \( w_2 \) end in \( v \).
	Thus also at every vertex \( v \in V_{\geq 3} \) at most one walk ends.
	Then there are at most \( k^2 + 2k \) endpoints of walks, and thus there are at most \( \lfloor (k^2 + 2k)/2 \rfloor \leq k^2 \) walks in \( \mathcal{W} \).
%	Then there are at most \( k^2 + 2k \) endpoints of walks, and thus there are at most \( (k^2 + 2k)/2 \leq 2k^2 \) walks in \( \mathcal{W} \).
	
	Consider a walk \( w \in \mathcal{W} \) where a vertex \( v \) occurs more than once.
	Recall that the embedding of a path \( p \in P \) of a yes-instance is injective, thus no vertex \( v \in V(G) \) occurs twice in the same path.
	A naive approach would be to now solve the bin-packing problem of `weights' \( P \) and `bins' \( \mathcal{W} \).
	Then, however, a solution to the bin-packing would may potentially translate to an embedding of a path where a vertex occurs twice.
	Therefore let us guess a partition into paths without multiple occurrence of vertices, as follows.
	
	Between two occurrences of \( v \) on walk \( w \) there must be vertex \( u \) (possibly an occurrence of \( v \) itself) which is the endpoint of two different paths.
	Therefore there is a partition of the walks \( \mathcal{W} \) into vertex-disjoint paths \( \mathcal{X} \), where still paths \( P \) have an embedding into \( \mathcal{X} \).
	We may describe this partition by `cuts' of \( \mathcal{W} \) specified by a vertex \( v \) in the union of walks from \( \mathcal{W} \).
		Note, that in the union of walks \( \mathcal{W} \), each high degree vertex \( v \in V_{\geq 3} \setminus V_2^\star \) occurs \( \mbox{deg}(v) \leq k \) times.
	Thus there are to up to \( n + k^2 \) potential cut vertices.
	
	We claim that at most \( k^2 \) cuts \( C \) are necessary to cut the walks \( \mathcal{W} \) into vertex-disjoint paths \( \mathcal{X} \).
	Assume, that there is a cut vertex \( v \in C \) which is on an inner vertex of a path between \( V_1 \) and \( V_{\geq 3} \).
	Then joining its incident vertex-disjoint paths results in a vertex-disjoint path.
	Thus we may assume that the cuts \( C \) are at vertices from walks of \( \mathcal{Q} \) between vertices among \( V_{\geq 3} \).
	Let \( \mathcal{Q}_{\geq 3} \) be the set of paths \( \mathcal{Q} \) with endpoints in \( V_{\geq 3} \).
	Consider the multi-graph with loops on vertex set \( V_{\geq 3} \) with an edge between \( u,v \in V_{\geq 3} \) for every path \( \mathcal{Q}_{\geq 3} \) with endpoints \( u \) and \( v \).
	Since \( |Q_{\geq 3}| \leq k \) and the degree of every vertex \( v \in Q_{\geq 3} \) is at most \( k \), this multi-graph has at most \( k^2 \) edges.
	Then also there are at most \( k^2 \) paths \( \mathcal{Q}_{\geq 3} \).
	Consider a set of more than \( k^2 \) vertices \( C \subseteq V \) that cut \( \mathcal{Q} \) into vertex-disjoint paths \( \mathcal{X} \).
	Then there is a path of \( Q_{\geq 3} \) containing distinct cut vertices \( u,v \in C \).
	Let \( x \in \mathcal{X} \) be the path between \( u \) and \( v \).
	Joining them with the incident path at, say \( u \), results in a vertex-disjoint path.
	Thus cutting \( \mathcal{W} \) at vertices \( C \setminus \{u\} \) still results in a set of vertex-disjoint paths.
	Therefore at most \( k^2 \) cuts of walks \( \mathcal{W} \) are necessary to yield vertex-disjoint paths \( \mathcal{X} \).
	
	Let us utilize this observation in the design of our algorithm.
	Guess up to \( k^2 \) cuts \( C \) from the \( n + k^2 \) potential cuts.
	Cut the previously guessed walks \( \mathcal{W} \) into subpaths according to cuts \( C \).
		We may force exactly \( k^2 \) cut vertices by allowing \( C \) to be a multi-set containing also leaves, whose cut has no effect.
	This way, we try another \( (n+k^2)^{k^2} \) possibilities of cut vertices \( C \).
	Then cut the previously guessed \( \leq k^2 \) walks \( \mathcal{X} \) at the \( k^2 \) cut positions.
	If the resulting set of walks is not vertex-disjoint, discard this guess.
	Otherwise we obtain \( \leq 2 k^2 \) vertex disjoint paths \( \mathcal{X} \), since every cut increases the number of paths by one.
	This resembles a bin-packing problem in unary encoding with \( k^2 \) bins of different sizes and total capacity \( n \).
	We may apply standard dynamic programming technique to test in \( n^{O(k^2)} \) time whether the sizes of the paths \( P \) fit into the bins in the sizes of \( \mathcal{X} \).
	If the paths \( P \) fit in some guessed paths \( \mathcal{X} \), then corresponding partition of the edges in \( G \) yields paths \( P \).
		Thus there is an edge-disjoint embedding of \( P \) into \( G \).
	For the other direction, if the edges of \( G \) can be partitioned into paths \( P \), then as argued before there is a set \( \mathcal{X} \) according to this partition and the there is a solution to the dynamic problem.
	The runtime is \( k!^k (n+k^2)^{O(k^2)} \cdot \mbox{poly}(n) = k!^k n^{O(k^2)} \cdot \mbox{poly}(n) \) where \( \mbox{poly} \) is a polynomial.
\end{proof}
%}

Can we achieve a better runtime than \( k!^k n^{k^2+O(1)} \), in particular decrease the dependence on \( k \) in the exponent of \( n \)?
Not significantly unless ETH fails, as the following reduction from \mwnp shows.

\begin{theorem}\label{the:pp-lower-bound-fpt}
	There is no algorithm that decides \PP in time $ n^{o({k^2}/{\log k})} $ with \( k = \bcd(G) \) unless ETH fails.
\end{theorem}
\short{}{
\begin{proof}
    We give an \fpt reduction from \mwnp parameterized by the number of sets.
    Assume we have an instance $I$ with $k$ sets. 
    Without loss of generality, we can assume $\sqrt{k} \in \N$.
    We construct an equivalent instance of the \pp problem of polynomial size
    with $\sqrt{k}$ components, maximal degree $2\sqrt{k}$ and $\sqrt{k}$
    vertices of degree $\geq 3$.
    By \cref{lem:multiwayhard}, this is sufficient.
	
	Each weight $w_j$ of $I$ becomes a path $p_j$ with length $|p_j| = w_j\cdot 2k$. 
	We construct $\sqrt{k}$ subdivided stars, each with degree $2\sqrt{k}$. 
	Each leg of a subdivided star has length $\frac{1}{2k}\sum_{j=1}^n |p_j|$. 
	%Two legs in the same component represent one set of $I$.
	
    Let $S_1,\dots,S_k$ be a feasible solution for $I$.
    The sum of weights for each set $S_i$
    is $\sum_{j\in S_i} w_j = \frac{1}{k} \sum_{j=1}^n w_j$. 
    The corresponding paths can be packed exactly into two legs of a subdivided
    star because they have a summed length of
    \[
    \sum_{j\in S_i} |p_j|= 2k
    \sum_{j\in S_i} w_j = 2 \sum_{j=1}^n w_j = \frac{1}{k}\sum_{j=1}^n |p_j|.
    \]
	
	A solution for the \pp instance can easily be transferred to a solution for \mwnp. 
	Each pair of two legs in a subdivided star is one set. 
    Identifying such a pair is easy: A path is either part of both
    legs because it contains the center vertex of the star or we choose two
    legs where the paths end at the center vertex. 
	The number of these legs must be even. 
	
	Because we have $\sqrt{k}$ subdivided stars with degree $2\sqrt{k}$, we can construct $k$ leg pairs which represent the sets. 
	Thus, the parameter in our \fpt reduction becomes quadratic. 
	If the \mwnp can not be solved in $O(n^{{k}/{\log k}})$, the \pp problem, parameterized by the number of components, maximum degree and vertices of high degree, cannot be solved in $O(n^{{k^2}/{\log k}})$. 
\end{proof}
}

\section{Conclusion}

We showed that edge-disjoint packing of paths into a graph is a
very hard problem.  Even if the input graph is a subdivided star or a
linear forest the problem is hard. 
If we parameterize the problem by the number of paths,
the problem remains hard even
for input graphs with treewidth two.  
However, it becomes fixed parameter tractable on forests.
A natural open problem 
is to not embed paths, but more general graphs such as trees or cycles.

\bibliographystyle{plainurl}% the mandatory bibstyle
\bibliography{references}

\end{document}